\documentclass[online]{tlp}

\usepackage{amsfonts}
\usepackage{todonotes}

\usepackage{booktabs}
\usepackage{bbm}
\usepackage{amsmath}
\usepackage{graphicx}
\usepackage{multirow}
\usepackage{listings}
\usepackage{tikz}
\usepackage{siunitx}
\usepackage{mathdots}
\usepackage{subcaption}
\usepackage{csquotes}
\usepackage[inline]{enumitem}
\usepackage{xifthen}
\usepackage{hyperref}
\usepackage{cleveref}

\renewcommand{\arraystretch}{1.2}

\usetikzlibrary{automata,positioning,shapes.geometric,shapes.callouts,trees,patterns}

\newcommand*{\indic}[1]{\mathbbm{1}_{#1}}
\newcommand*{\expect}[1]{\ensuremath\mathsf{E}[#1]}
\newcommand*{\prefix}{\ensuremath\prec}

\newcommand*{\prog}[1]{\ensuremath\mathcal{#1}}
\newcommand*{\naturals}{\mathbbm{N}}

\newcommand*{\set}[1]{\{#1\}}
\newcommand*{\block}[1]{\mathsf{Block}(#1)}
\newcommand*{\mapping}[3]{#1\colon #2 \rightarrow #3}
\newcommand*{\defeq}{\stackrel{\text{\tiny def}}{=}}
\newcommand*{\floor}[1]{\lfloor #1 \rfloor}
\newcommand*{\bigTheta}{\ensuremath\Theta}
\newcommand*{\bigOmega}{\ensuremath\Omega}
\newcommand*{\bigO}{\ensuremath \mathcal{O}}
\newcommand*{\differential}[1]{\frac{\mathrm{d}}{\mathrm{d}#1}}
\newcommand*{\Init}{\ensuremath\mathsf{Init}}
\newcommand*{\Exit}{\mathsf{Exit}}
\newcommand*{\HT}[2][]{%
    \ensuremath \ifthenelse{\isempty{#1}}{H^{#2}}{H_{#1}^{#2}}%
}
\newcommand*{\MHT}[2][]{%
    \ensuremath \ifthenelse{\isempty{#1}}{h^{#2}}{h_{#1}^{#2}}%
}
\newcommand*{\HTr}[2][]{%
        \ensuremath \ifthenelse{\isempty{#1}}{R^{#2}}{R_{#1}^{#2}}%
    }
\newcommand*{\MHTr}[2][]{%
    \ensuremath \ifthenelse{\isempty{#1}}{r^{#2}}{r_{#1}^{#2}}%
}

\newcommand*{\states}{\mathcal{S}}
\newcommand*{\statesDS}{\mathcal{S}}
\newcommand*{\mkvDS}{\mathcal{M}}
\newcommand*{\state}[1]{\langle #1 \rangle}
\newcommand*{\shorten}{\mathsf{Pop}}

\newtheorem{proposition}{Proposition}
\newtheorem{theorem}{Theorem}
\newtheorem{lemma}{Lemma}
\newtheorem{corollary}{Corollary}

\newtheorem{remark}{Remark}
\newtheorem{fact}{Fact}

\begin{document}

\lefttitle{M. Gelderie, M. Luff, M. Peltzer}

\jnlPage{1}{8}
\jnlDoiYr{2021}
\doival{10.1017/xxxxx}

\title[Theory and Practice of Logic Programming]{Impact and Performance of Randomized Test-Generation using Prolog\thanks{This work was
        created as part of a project funded by the German Federal Ministry of Education and Research
        under grant number 16KIS1913K. The authors are responsible for the content of this
publication.}\thanks{This paper is an extended version of \cite{GeLuPe24}, which was published in
the 34th International Symposium on Logic-Based Program Synthesis and Transformation. We thank
the PC chairs Juliana Bowles and Harald Søndergaard.}}

\begin{authgrp}
\author{\gn{Marcus} \sn{Gelderie}}
\affiliation{Aalen University of Applied Sciences\\
    Beethovenstr.\ 1, 73430 Aalen, Germany\\
    \email{marcus.gelderie@hs-aalen.de}
}
\author{\gn{Maximilian} \sn{Luff}}
\affiliation{Aalen University of Applied Sciences\\
    Beethovenstr.\ 1, 73430 Aalen, Germany\\
    \email{maximilian.luff@hs-aalen.de}
}
\author{\gn{Maximilian} \sn{Peltzer}}
\affiliation{Aalen University of Applied Sciences\\
    Beethovenstr.\ 1, 73430 Aalen, Germany\\
    \email{maximilian.peltzer@hs-aalen.de}
}
\end{authgrp}


\maketitle

\begin{abstract}
    We study randomized generation of sequences of test-inputs to a system using Prolog. Prolog is a
    natural fit to generate test-sequences that have complex logical inter-dependent structure. To
    counter the problems posed by a large (or infinite) set of possible tests, randomization is a
    natural choice. We study the impact that randomization in conjunction with SLD resolution have
    on the test performance. To this end, this paper proposes two strategies to add randomization to
    a test-generating program. One strategy works on top of standard Prolog semantics, whereas the
    other alters the SLD selection function. We analyze the mean time to reach a test-case, and the
    mean number of generated test-cases in the framework of Markov chains. Finally, we provide an
    additional empirical evaluation and comparison between both approaches.

    Under consideration in Theory and Practice of Logic Programming (TPLP).
\end{abstract}


\section{Introduction}

The need for software testing is well established. The idea to auto-generate tests is a constant
theme in the field of software testing, stretching back many decades (e.g.\
\cite{miller1975automated,ince87,pesch1985test,meyer07OOPTesting}). Automatically generating
software tests can be done in a number of ways, depending on the specific test-goal: In the past,
tests have been generated from UML specifications \cite{kim1999test}, based on natural language
\cite{xu22NaturalLanguage}, and, more recently, using large-language models
\cite{gu2023llm,siddiq2023exploring}. But tests have also been generated according to formal or
semi-formal specifications \cite{zeng2002functional,Dewey14CLP}. Particularly when formal methods
are used, one often has to deal with a very large, even infinite, number of test cases. Exploring
such a large set of tests in a randomized fashion is a natural approach and has been used
extensively in various different ways
and contexts for a long time (see e.g.\ \cite{ramler2012random,Miller90,duran1984evaluation,%
ramler2012random,Miller90,casso2019integrated}).

Prolog is a natural fit to generate test-cases that follow a logical pattern (as opposed to
unstructured testing, as is done, for example, in many forms of Fuzzing \cite{Miller90}). Generating
test-cases using Prolog has been studied in the past
\cite{pesch1985test,hoffman1991automated,casso2019integrated,denney1991test}. It has been applied to software-testing in
general, but also to specialized areas, such as security testing \cite{zech2019knowledge,Dewey14CLP}. Some
approaches also use randomization to explore the space of test-cases \cite{casso2019integrated}.
Randomization solves some of the problems inherent in the SLD resolution algorithm --- particularly
the fact that it is not complete when the resolution works in a depth-first manner. It may also
yield a more diverse set of test-cases, because it permits exploring distant parts of an infinite
SLD tree. Randomization seems to be a logical fit in the context of test-case generation using
Prolog.

In the light of its apparent utility, it is natural to study randomization itself and its
properties. What are the possible strategies to implement randomized search strategies for test-cases in
Prolog running on current state-of-the-art implementations? What is the probability of hitting a
particular test case, and how long will it take? To our surprise, we only found very few papers
dealing with the properties of randomization itself (see also \emph{Related Work} below).

In this paper, we study randomized test-case generation using Prolog. Our main contributions are
threefold: \begin{enumerate*}[label=(\roman*)]\item  We propose strategies to implement randomized
    search in both unmodified Prolog
    runtimes, and via specific modifications to the usual SLD implementations. \item We show how
    adding randomness naturally turns the SLD resolution into an infinite discrete-time Markov chain
    and propose to use this framework to study the runtime-effects. We do this for our proposed
    scheme and give tight asymptotic bounds on the expected time to hit a particular test-case.
    \item Finally, we study the effect that various Prolog implementations have on the efficiency of
randomizing test-case generation.\end{enumerate*}

We present two ways of adding randomization to Prolog programs. The first way works without altering
the semantics of standard Prolog and thus works on existing implementations. It works by adding a
predicate, called a \emph{guard}, that randomly fails to every rule. Crucially, failure is
determined by an independent event for every successive call to the same rule. We refer to this
strategy as the \emph{guard approach}. In a second strategy,
we propose a modification to the resolution algorithm: Given a goal and a set of matching rules,
drop an indeterminate number of rules from the set and permute the remaining ones. Again, we do this
in an independent fashion every time a goal is resolved with the input program. This second
modification is reminiscent of that proposed in \cite{casso2019integrated}, but differs in that it
also \emph{drops} a random number of rules from the set. This, in effect, prevents an
infinite recursion with probability 1. We refer to the second strategy as the
\emph{drop-and-shuffle approach}.

In the following we study the effects on randomizing the resolution in this way. We give a detailed
description of the resulting Markov chain and analyze its probability structure. We show that,
provided the parameters are chosen appropriately, the number of test cases produced is finite and
given by a simple equation in terms of the selected probabilities. This is true for both approaches
to randomization. We also show that, if we repeat the initial query infinitely many times, we will
reach each test-case after a finite number of steps on average. This \emph{hitting time} is a
well-known concept in the study of Markov chains. We again give a closed formula representation and
accompanying asymptotic bound in the depth of the given test case in the SLD resolution tree. Again,
this is done for both approaches, though the drop-and shuffle approach permits for a narrower set of
parameters to ensure the computed quantities are finite.

Finally, we study the randomization procedures from an empirical perspective and provide comparisons
between the two aforementioned approaches. We implement the guard approach to randomization
in SWI-Prolog \cite{swi} and the drop-and-shuffle approach in Go-Prolog \cite{go_prolog}. We
chose Go-Prolog for its accessible and simple code-base, which lends itself to experimental
modifications. We then compare the number of test-cases produced before a specific test-goal is seen
and the number of iterations that were required to do so.


\subsection{Related Work}

Some early works on test-case generation using Prolog are
\cite{pesch1985test,bouge1985application,hoffman1991automated,denney1991test}. Automated test-case
generation in Prolog was described by \cite{pesch1985test}. The authors state how
to test individual syscalls with logic programming. The used specifications state a set of
pre-conditions then the actual invocation of the respective syscall and afterwards what the expected
post-conditions are. This paper demonstrates that test-case generation using Prolog is very
beneficial to test systems in a structured manner. This problem domain does not deal with any
problems of recursion since the authors only test input sequences of length one. This means that
this paper does not deal with recursion problems that are witnessed for many other test scenarios.

Another approach showcasing Prolog's capabilities used in test case generation
was shown by \cite{hoffman1991automated}. The authors
automated the generation procedure of tests for modules written in C with
Prolog.

\cite{bouge1985application} start the testing procedure with the definition of a
$\Sigma$-algebra and respective axioms. The aim of this testing procedure is based on the regularity
and uniformity testing hypothesis. Prolog is used to generate test cases and to partition the test
cases into test classes following the uniformity hypothesis. The authors also recognize the problems
of recursion in Prolog test case generation and apply different search strategies to solve them.
Since the paper enforces a length limit on the generated solution, it will not find any test case
that exceeds that length.

\cite{denney1991test} also researches test-case generation based
on specifications written in Prolog. In his paper, he implements a
meta-interpreter in Prolog to be able to track which rules, generated from the
specification, were already applied. This is done by constructing a finite
automaton. Each arc between states corresponds with respective rules in the
Prolog database. Final states in this automaton are test cases produced in the
test-case generation process. With this solution, he addresses the problems of
recursion, evaluable predicates, and ordering, which are challenging aspects of
test-case generation using Prolog. However, the recursion problem is only
addressed heuristically, which means that a user has to specify a threshold of
how often an arc can be traversed during the execution of test case generator.
We argue that the estimation of the threshold is an error-prone task and, if
not set correctly, could miss important test cases.

\cite{gorlick1990mockingbird} also introduce a methodology for
formal specifications. For this task, they use constraint logic programming to
describe the system under test's behaviour. With this approach, the authors
also recognize that they both have a test oracle and a test case generator at
the same time. One challenge the authors addressed is, yet again, the recursion
problem. To solve this challenge they used a randomization approach. This
feature enables the proposed framework to pick probabilistically from the
predicates. However, they do not provide any statements about test case
duplication or infinite looping.

\cite{casso2019integrated} approach assertion-based testing of
Prolog programs with random search rules. They rely on the \texttt{Ciao} model
and its capabilities to specify pre- and post-conditions for static analysis
and the runtime checker. Further, the authors develop a test-case generator
based on these conditions. For randomizing the test case search, Casso et.\ al.\
use a selection function that randomly chooses clauses to be resolved. The authors
do not study the randomization itself, nor its properties. We will revisit this
paper and its randomization strategy in \cref{sec:strategies}, where we will also
explain the differences from our approach in more detail.

Prolog was also used in security testing. For web applications,
\cite{zech2014security,zech2019knowledge} first build an expert
system to filter test cases according to some attack pattern and later apply
this risk analysis to filter test cases in the generation process. Since the
paper, yet again, only addresses single input sequences, it effectively
circumvents the problem of recursion. Prolog was also used in Fuzzing by \cite{Dewey14CLP} to use CLP in order to produce fuzzing inputs to compilers.

\section{Preliminaries}

Given a (usually finite) set $\Sigma$ of elements, we write $\Sigma^*$ for the set of all finite
length sequences $w_1\cdots w_l$ with $w_i\in\Sigma$ and
$l\in\naturals_0=\set{0,1,2,3,\ldots}=\naturals\cup\set{0}$. The empty sequence is
denoted by $\varepsilon$. We write $\Sigma^+=\Sigma^*\setminus\set{\varepsilon}$. If
$\Sigma=\set{x}$ is a singleton, we write $x^*$ or $x^+$ instead of $\set{x}^*$. Concatenation is
denoted by $(u_1\cdots u_l)\cdot (v_1\cdots v_r)=u_1\cdots u_lv_1\cdots v_r$. We write
$|w|=|w_1\cdots w_l|=l\in\naturals_0$.

We use the theory of \emph{Markov chains}. For a detailed introduction and proofs of the following
claims, the reader is referred to standard literature on the subject, e.g.\ \cite{norris1998markov}.
We revisit the concepts, notation, and central results from the theory of Markov chains that we will
use throughout this paper for convenience.

We consider a countable set $\states$ of \emph{states}, a mapping
$\mapping{p}{\states\times\states}{[0,1]}$ that assigns \emph{transition probabilities} to pairs of
states with the property that  for all $s\in\states$ it holds that $\sum_{s'\in\states}p(s,s')=1$,
and an \emph{initial state}\footnote{In the literature one usually considers \emph{initial
    distributions} to model uncertainty about the initial state. We do not need this capability in
    the present paper and consider initial distributions whose support is a single state of
probability 1.} $\Init\in\states$. Let $(X_n)_{n\in\naturals_0}$ be an infinite sequence of
random variables $X_n\in\states$. The tuple $(\states, (X_n)_{n\in\naturals_0}, p, \Init)$ is a
\emph{Markov chain}, if $\Pr[X_0=\Init]=1$ and for all $n\in\naturals$ and all
$s_1,\ldots,s_n\in\states$:
\begin{multline*}
    \Pr[X_n=s_n\mid X_0=s_0=\Init,\ldots, X_{n-1}=s_{n-1}]\\
    =\Pr[X_n=s_n\mid X_{n-1}=s_{n-1}] = p(s_{n-1},s_n)
\end{multline*}

For two states $s,s'$ we write $s\leadsto s'$, if $\Pr[X_n=s'\text{ for some }n ]>0$ in the Markov
chain $(\states, (X_n)_{n\in\naturals},p,s)$. Intuitively, there is a way to get from $s$ to $s'$. A set
$A\subseteq \states$ is \emph{absorbing}, if for every $s\in A$ and every $s'\in \states$ with
$s\leadsto s'$ it holds that $s'\in A$. If $A=\set{s}$ is a singleton, the state $s$ is said to be
absorbing. If any two states are reachable from one-another ($s\leadsto s'$ for any $s,s'\in\states$),
the Markov chain is \emph{irreducible}.

Let $A\subseteq \states$ be a non-empty set of states and let $\HT{A}=\inf\set{n\in\naturals_0\mid
X_n\in A}\in\naturals_0\cup\set{\infty}$ denote the random variable such that $X_H\in A$ visits $A$
for the first time. $H^A$ is the \emph{hitting time} of $A$. Then conditioned on $\HT{A}<\infty$ and
$X_H=s$, the sequence $(X_{H+n})_{n\in\naturals_0}$ is a Markov chain with initial state $s$ and is
independent of $X_0,\ldots,X_H$. This is called the \emph{strong Markov property}.
It is sometimes useful to consider the hitting times for initial states other than $\mathsf{Init}$.
Write $\HT[s]{A}$ for the hitting time of $A$ with starting state $s$.

The expected value $\MHT{A}\defeq\expect{\HT{A}}$ is known as the \emph{mean hitting time}. Given
any state $s\in\states$, we define $\MHT[s]{A}=\expect{\HT[s]{A}}$ for the mean hitting time of $A$
from initial state $s$. The mean hitting times are then the unique minimal positive solution to the
equations
\begin{equation}\label{eqn:mean_hitting_times}
    \begin{array}{ll}
        \MHT[s]{A}=0  &\quad\text{if $s\in A$}\\
        \MHT[s]{A}=1 + \sum_{s'\in\states} p(s,s')\MHT[s']{A}  &\quad\text{if $s\notin A$}
    \end{array}
\end{equation}

A state $s\in\states$ is \emph{recurrent}, if $\Pr[\sum_{n=0}^{\infty}\indic{X_n=s}=\infty]=1$ (where
$\indic{A}$ is the indicator random variable for event $A$).
Otherwise, it is \emph{transient}. It can be shown that a state is recurrent iff $\Pr[X_{m}=s\text{
for some } m\geq 1]=1$ in the chain $(\states, (X_n)_{n\in\naturals_0},p,s)$ (the probability of returning
$s$, once visited, is 1). One can show that if a Markov chain is irreducible and
contains one recurrent state, then all states are recurrent. In the case we call the chain itself
\emph{recurrent} (or \emph{transient}).

\section{Randomized Test Generation with Prolog}\label{sec:strategies}

In this paper, we view a test as a sequence of inputs to a system. For example, given a
web-application with a REST-interface, we could think of a test as a sequence of HTTP-Requests using
various methods (\texttt{GET}, \texttt{POST} and so forth) against different API-endpoints (e.g.\
\texttt{/login}, \texttt{/items/\{USERID\}/list}). Since our focus is on randomization, we do not
explicitly model a concept of \enquote{valid} test-cases. We also do not model the
\emph{test-oracle} which determines the success or failure of the test (e.g.\ \enquote{requests
are processed in $< \qty{700}{\ms}$}).

\begin{lstlisting}[caption={A program generating randomized sequences of test
inputs.},float,label={lst:generator_program}]
t([]).
t([H|T]) :- command(H), t(T).(*@\label{lst:generator_program:t_rule_cont}@*)
command(X) :- command1(X); /* ... */ ; commandr(X).
command1(X) :- /* ... */ .
% ...
\end{lstlisting}

\begin{lstlisting}[caption={Guard clauses.},emph={guard_t,guard_1,guard_r},emphstyle={\itshape\color{blue!80!black}},float,label={lst:generator_guards}]
guard_t :- random(X), X < p_cont.
guard_1 :- random(X), X < p_1.
% ...
t([H|T]) :- guard_t, command(H), t(T).
command1(X) :- guard_1, /* ... */ .
% ...
\end{lstlisting}

At a very abstract level, such a sequence of test-inputs could be generated with the Prolog program
shown in \cref{lst:generator_program}. All valid substitutions for \texttt{X} in the query
\texttt{t(X)} are input sequences to our fictional system. Since this program will only ever output
test sequences of the type \texttt{[command1, command1, command1, \ldots]}, a straightforward
approach is to add \emph{guard clauses} of the form shown in \cref{lst:generator_guards}. Note that
the symbols \texttt{p\_cont}, \texttt{p\_1},\ldots are meant to represent float constants between 0
and 1, and can be adjusted as needed. In effect, some sub-trees of the SLD-tree are then randomly
left unexplored. We refer to this as the \emph{guard approach}.

Another, superficially similar strategy was proposed by \cite{casso2019integrated}. Their randomization
is presented as a modification to the Prolog interpreter; equivalently, it can be implemented using
meta-predicates. Essentially, Casso et.\ al.\ shuffle the list of input clauses whose head unifies
with the current goal, instead of iterating over it in the usual left-to-right
fashion. They do not drop rules. The termination of the program is instead enforced via
depth-control. It is thus not difficult to see that the random approach itself merely
alters the order of test-cases, but not their number. As such, the questions concerning the number
of test-cases (that we study here) do not make sense for their approach.

However, one can augment the shuffling approach due to Casso et.\ al.\ by \emph{additionally}
dropping several items from the set of unifying rules prior to shuffling. We do this with an
independent Bernoulli trial for each rule (i.e. the number of dropped rules follows a Binomial). The
resulting algorithm shares many properties with our scheme above (in particular, the results from
the next section apply). We refer to this approach as the \emph{drop-and-shuffle}
strategy. We proceed to study both approaches below.

\subsection{Guard Strategy}\label{sec:guard}

\subsubsection{Number of Generated Tests}\label{subsec:guard}

The program $\prog{P}$ shown in
\cref{lst:generator_program,lst:generator_guards} gives rise to a probabilistic number of test
cases. We study the questions: Is this number finite? If so, what is the expected number of
test-case?

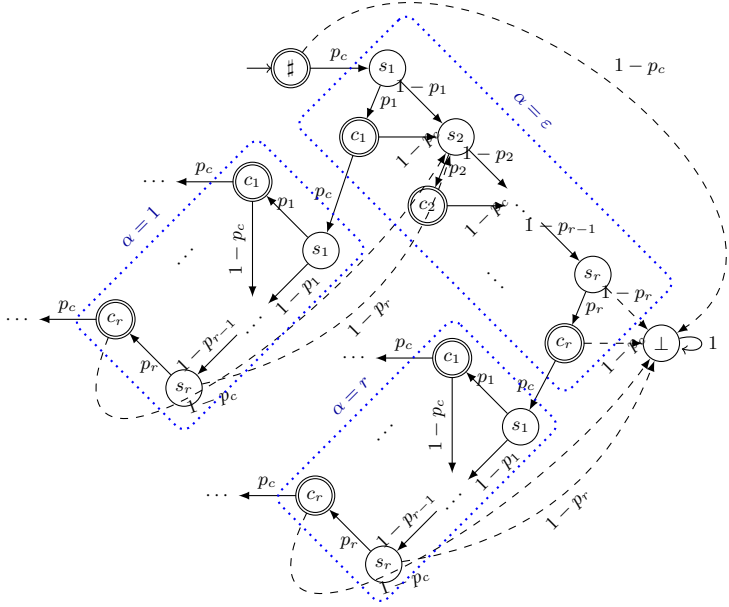
\begin{figure}[t]
    \centering
\begin{tikzpicture}[initial text=,every state/.style={inner sep=0pt,font=\small,minimum
    size=.6cm},node distance=1.6cm,on grid, auto, transform shape, scale=.8]

    \clip (-4.8,1.6) rectangle (8, -9);

    \node[state,initial,accepting] (init) at (0,0) {$\sharp$};
    \node[state] (sele01) [right=of init] {$s_1$};
    \node[state] (sele02) [below right=of sele01] {$s_2$};
    \node[rotate=-6] (dots) [below right=of sele02] {$\ddots$};
    \node[state] (sele0r) [below right=of dots] {$s_r$};
    \node[state,accepting] (q01) [left=of sele02] {$c_1$};
    \node[state,accepting] (q02) [left=of dots] {$c_2$};
    \node[rotate=-6] (lodots) [below right=of q02] {$\ddots$};

    \node[state] (dead) [below right=of sele0r] {$\bot$};
    \node[state,accepting] (q0r) [left=of dead] {$c_r$};

    \path[-latex]
        (init) edge node[above]{$p_c$} (sele01)
               edge[bend left,out=90,looseness=1.2,in=70,dashed] node {$1-p_c$} (dead)
        (sele01) edge node[above]{$1-p_1$} (sele02)
        edge node[right]{$p_1 $} (q01)
        (sele02) edge node[above]{$1-p_2$} (dots)
        edge node[right]{$p_2 $} (q02)
        (dots) edge node[above]{$1-p_{r-1}$} (sele0r)
        (sele0r) edge node[right]{$p_r $} (q0r)

        (q01) edge node[rotate=45,below] {$1-p_c$} (sele02)
        (q02) edge node[rotate=45,below] {$1-p_c$} (dots)
        (dead) edge[loop right] node {$1$} (dead)

        ;
    \path[dashed,-latex]
        (q0r) edge node[rotate=45,below] {$1-p_c$} (dead)
        (sele0r) edge node[above]{$1-p_r$} (dead)
        ;

    \node[text={black!40!blue},rotate=-45] at (4.0,-.7) {$\alpha=\varepsilon$};
    \node[text={black!40!blue},rotate=45] at (-2.5,-2.6) {$\alpha=1$};
    \node[text={black!40!blue},rotate=45] at (1,-5.4) {$\alpha=r$};

    \node[state] (sele11) at (.5,-3) {$s_1$};
    \node[state,accepting] (q11) [above left=of sele11] {$c_1$};
    \node (onw21) [left=of q11] {$\cdots$};

    \node[rotate=10] (dots12) [below left=of sele11] {$\iddots$};
    \node[rotate=10] (ldots12) [below left=of q11] {$\iddots$};

    \node[state] (sele1r) [below left=of dots12] {$s_r$};
    \node[state,accepting] (q1r) [above left=of sele1r] {$c_r$};
    \node (onw2r) [left=of q1r] {$\cdots$};

    \path[-latex]
        (q01) edge node[left]{$p_c$} (sele11)
        (sele11) edge node[below,rotate=45]{$1-p_1$} (dots12)
        edge node[above]{$p_1 $} (q11)
        (dots12) edge node[rotate=45,above]{$1-p_{r-1}$} (sele1r)
        (sele1r) edge node[below]{$p_r $} (q1r)

        (q11) edge node[rotate=90,above] {$1-p_c$} (dots12)

              edge node[above] {$p_c$} (onw21)
        (q1r) edge node[above] {$p_c$} (onw2r)

        ;
        \path[dashed,-latex]
            (sele1r) edge[bend right] node[below,rotate=45]{$1-p_r$} (sele02)
        (q1r) edge[looseness=1.4,bend right,out=220] node[below,rotate=20] {$1-p_c$} (sele02)
        ;

\draw[rounded corners,draw=blue,thick,dotted,rotate=-45] (.7,2.0) rectangle (7.0,-.6);
\draw[rounded corners,draw=blue,thick,dotted,rotate=45] (-1.2,-.4) rectangle (-5.5,-3);

    \node[state] (seler1) at (3.8,-5.9) {$s_1$};
    \node[state,accepting] (qr1) [above left=of seler1] {$c_1$};
    \node (onwr1) [left=of qr1] {$\cdots$};

    \node[rotate=10] (dotsr2) [below left=of seler1] {$\iddots$};
    \node[rotate=10] (ldotsr2) [below left=of qr1] {$\iddots$};

    \node[state] (sele1r) [below left=of dotsr2] {$s_r$};
    \node[state,accepting] (qrr) [above left=of sele1r] {$c_r$};
    \node (onwrr) [left=of qrr] {$\cdots$};

    \path[-latex]
        (q0r) edge node[left]{$p_c$} (seler1)
        (seler1) edge node[below,rotate=45]{$1-p_1$} (dotsr2)
        edge node[above]{$p_1 $} (qr1)
        (dotsr2) edge node[rotate=45,above]{$1-p_{r-1}$} (sele1r)
        (sele1r) edge node[below]{$p_r $} (qrr)

        (qr1) edge node[rotate=90,above] {$1-p_c$} (dotsr2)

              edge node[above] {$p_c$} (onwr1)
        (qrr) edge node[above] {$p_c$} (onwrr)

        ;
    \path[dashed,-latex]
        (sele1r) edge[bend right] node[below,rotate=45]{$1-p_r$} (dead)
        (qrr) edge[looseness=1.4,bend right,out=220] node[below,rotate=20] {$1-p_c$} (dead)
    ;
\draw[rounded corners,draw=blue,thick,dotted,rotate=45] (-1,-4.8) rectangle (-5.2,-7.3);

\end{tikzpicture}
\caption{The Markov chain corresponding to $\prog{P}$ with $\Init=\sharp$ and \emph{blocks}
$\alpha\in\set{1,\ldots,r}^*$ surrounded by blue boxes.}
\label{fig:chain}
\end{figure}

The program $\prog{P}$ gives rise to an infinite Markov chain, which is based on the SLD-tree
corresponding to $\prog{P}$. Recall that $\mathcal{P}$ is governed by some probabilities
\texttt{p\_{cont}}, \texttt{p\_1}, \ldots, which we will denote by $p_c$, $p_1,\ldots, p_r$. The
Markov chain is depicted in \cref{fig:chain}. Note that we model choice points via the states $s_i$.
This is necessary, because $\mathcal{P}$ will backtrack when a call to \texttt{command1} fails, and
proceed to \texttt{command2} with probability $p_2$. Double-circles denote \emph{output states} ---
that is, whenever such a state is visited, a test case terminating in that state is generated. Node
$\sharp$ corresponds to the empty list. $\bot$ is the only absorbing state. It corresponds a
termination of the resolution algorithm.

The blue boxes denote areas that share a common structure. We call these areas \emph{blocks}. We can
uniquely identify each block by a finite sequence $\alpha\in\set{1,\ldots,r}^*$. For any state $s$ in
the Markov chain, we denote by $\block{x}$ the unique block that contains it. For any label
occurring in a block ($s_1,s_2,\ldots,s_r$ and $c_1,\ldots,c_r$) and a block $\alpha$, write
$s_1^\alpha$, $c_1^\alpha$, \ldots for the unique state with that label in block $\alpha$. In this
way, we can identify any state in the Markov chain. Put differently, the Markov chain is given by a
state space $\mathcal{S}=\set{s_i^\alpha,c_i^\alpha\mid 1 \leq i \leq r, \;
\alpha\in\set{1,\ldots,r}^*}\cup\set{\bot,\sharp}$ and transition probabilities $p(s,s')$ for
$s,s'\in\mathcal{S}$:
\begin{align*}
    p(\sharp,s_1^{\varepsilon})&=p_c\\
    p(\sharp,\bot)&=1-p_c\\
    p(s_i^\alpha,s_{i+1}^\alpha)&=1-p_i && 1\leq i < r\\
    p(c_i^\alpha,s_{i+1}^{\alpha})&=1-p_c && 1\leq i < r\\
    p(c_i^\alpha,s_{1}^{\alpha\cdot i})&=p_c && 1\leq i \leq r\\
    p(s_i^\alpha,c_{i}^\alpha)&=p_i && 1\leq i \leq r
\end{align*}
The dashed \emph{upward} arrows (which correspond to backtracking to a lower-recursion level) are
somewhat more technical to define. Those arrows originate in states of the form $s_r^\alpha$ or
$c_r^\alpha$. There are several cases to consider:
\begin{enumerate}[label=\alph*)]
    \item $\alpha \in\set{1,\ldots,r}^*\cdot i$ for some $1\leq i<r$
    \item $\alpha\in\set{1,\ldots, r}^*\cdot i \cdot r^+$ for some $1\leq i < r$
    \item $\alpha\in r^*$
\end{enumerate}
This motivates the following transition probabilities
\begin{align*}
    p(s_r^{\alpha\cdot i},s_{i+1}^\alpha)&=1-p_r && 1\leq i < r
                                         & \text{case a)}\\
    p(c_r^{\alpha\cdot i},s_{i+1}^\alpha)&=1-p_c && 1\leq i < r
                                         & \text{case a)}\\
    p(s_r^{\alpha'\cdot i\cdot r\cdots r},s_{i+1}^{\alpha'})&=1-p_r && 1\leq i< r
                                                            & \text{case
                                                            b)}\\
    p(c_r^{\alpha'\cdot i\cdot r\cdots r},s_{i+1}^{\alpha'})&=1-p_c && 1\leq i<r
                                                            & \text{case
                                                            b)}\\
    p(s_r^{r\cdots r},\bot)&=1-p_r &&
                           & \text{case c)}\\
    p(c_r^{r\cdots r},\bot)&=1-p_c &&
                           & \text{case c)}
\end{align*}

We call edges from a block $\beta\cdot\alpha$ to a state in block $\beta$ or from any block to
$\bot$ an \emph{upward} edge. They correspond precisely to the dashed arrows in \cref{fig:chain}. If
the Markov chain follows such an edge, we say block $\beta\cdot\alpha$ is \emph{left upward} or that
the chain \emph{traverses upward} at that point. A block that has been left upward, is never visited
again.

It is immediate that every state is visited at most once. There are no two states that can be
reached from one-another. Note further that if we omit the dashed arrows and the state $\bot$, the
resulting graph structure is an infinite, finitely branching tree.
Yet, it is conceivable that the terminal state $\bot$ is never reached, because the sequence of states
visited from $\sharp$ is infinite. The following proposition shows that this is not the case,
provided $p_c< 1$.
\begin{proposition}\label{prop:exit_prob1}
    Let $s\in\mathcal{S}$ be any state. If $p_c < 1$, then all sequences originating in $s$
    eventually leave $\block{s}$ upward. In particular, $\bot$ is visited eventually.
\end{proposition}
\begin{proof}
    Let $\alpha=\block{s}$. It is sufficient to show the result for $s=s_1^\alpha$. We first study
    the special case that there is an infinite path that \emph{never} traverses upward. Pick an
    infinite path $s_0 s_1 s_2\cdots$ through the chain that never traverses upward. For every $n$,
    the prefix $s_0\cdots s_n$ must traverse at least $t_n\defeq 1+\floor{\frac{n}{2r}}$ edges of
    the form $(c_i^\beta,s_1^{\beta\cdot i})$ (for correspondingly many distinct blocks $\beta$).
    This is because inside a block, there are only $2r$ states and no cycles. Hence, the probability
    of such a prefix is at most $p_c^{t_n}$ which tends to 0 as $n\to\infty$. As a result, the
    probability of any path that never traverses upward is 0.

    Now for any $i$, consider the subtree of nodes below $c_{i}^\alpha$ that are visited. Since
    every node in the Markov chain can be visited at most once, the only option to remain in this
    tree indefinitely is for the tree to be infinite. However, the Markov chain is finitely
    branching. Therefore, the subtree of visited nodes below $c_{i}^\alpha$ is finitely branching. By
    König's lemma, this tree contains an infinite path and hence has probability 0.
\end{proof}

\begin{corollary}
    Let $\alpha$ be any block. The probability of reaching $\alpha$ from $s_1^\varepsilon$ (i.e. from
    the initial block) is:
    \begin{displaymath}
        \Pr[\HT{\alpha}<\infty] = p_c^{|\alpha|}\cdot\prod_{i=1}^{|\alpha|}p_{\alpha_i}<\infty
    \end{displaymath}\label{cor:finite_hitting_prob}

    Consequently, the probability of reaching $\alpha$ from $\sharp$ is
    $p_c^{|\alpha|+1}\cdot\prod_{i=1}^{|\alpha|}p_{\alpha_i}$.
\end{corollary}

Let $s\in\mathcal{S}$. We denote by $N(s)$ the random variable that counts the total number of
states visited from $s$ (including those in downstream blocks), before $\block{s}$ is left upward.
A useful observation is that $N(s)=\HT[s]{E}$ can
also be expressed as a hitting time, where $E=\set{s_i^\beta\mid \beta \prefix\block{s},
\; 1\leq i\leq r}\cup\set{\bot}$. Note that it would suffice to take the subset of $E$ which
contains $s_{\alpha_i+1}^\beta$ for any $\beta=\alpha_1\cdots\alpha_{i-1} \prefix\alpha$. To define
this set, we would have to work around the case $\alpha_i=r$ --- indeed, if $\alpha\in r^*$, then
$E=\set{\bot}$. So we define $E$ as larger than needed purely to simplify notation.
Note moreover that $E$ depends on $\alpha=\block{s}$. Since $\alpha$ is usually clear from context,
we simply write $E$, but also use the notation $E_\alpha$ when needed.

\begin{lemma}
    Let $s\in\mathcal{S}$ and write $p_{\mathsf{max}}=\max\set{p_1,\ldots,p_r}$. If $p_c< 1$ and
    $\eta\defeq r\cdot p_{\mathsf{max}}\cdot p_c < 1$, then $\expect{N(s)}$ is
    finite.\label{lem:exp_finite}
\end{lemma}
\begin{proof}
    Let $\alpha=\block{s}$. It is obvious that $N(x^\alpha)\leq N(s_1^\alpha)$ for all
    $x\in\mathcal{S}$ with $\block{x}=\alpha$. It therefore suffices to show that $N(s_1^\alpha)$ is
    finite. In the remainder of this proof, we write $\hat{s}=s_1^\alpha$.

    Let now $\beta$ be any block and let $M_\beta$ denote the number of states visited in block
    $\beta$ from $s_1^\beta$. Clearly $M_\beta\leq 2r$. Let furthermore
    $I_\beta=\indic{\HT[\hat{s}]{\beta}<\infty}$ denote the indicator random-variable of the event
    that $\beta$ is visited from $\hat{s}$. Note that both random-variables are independent because
    the underlying random events in $\prog{P}$ are independent and we count by $M_\beta$ only states
    that are visited once $\beta$ is entered. We have:
    \begin{displaymath}
        N(\hat{s})=\sum_{\beta\in\alpha\cdot\set{1,\ldots,r}^*}I_\beta\cdot M_\beta
    \end{displaymath}

    There are precisely $r^l$ blocks that have distance
    $l\in\naturals$ from $\alpha$. For each such block $\beta=\alpha\cdot\beta_1\cdots \beta_l$, the probability of
    reaching it from $\alpha$ is $\Pr[I_\beta=1]=p_c^{l}\cdot\prod_{i=1}^lp_{\beta_i}$ by
    \cref{cor:finite_hitting_prob} (if $l=0$ then
    $\beta=\alpha$ and the probability is 1). Then $\Pr[I_\beta=1]\leq (p_{\mathsf{max}}\cdot p_c)^l$ for all $\beta$. This gives (using
linearity of expectation and that $I_\beta$ is independent from $M_\beta$ for all
$\beta$):
\begin{align*}
\expect{N(\hat{s})}&=\sum_{\beta\in\alpha\cdot\set{1,\ldots,r}^*}\expect{I_\beta}\cdot\expect{M_\beta}
                           \leq \sum_{l=0}^\infty r^l\cdot (p_c^l\cdot p_{\mathsf{max}}^l) \cdot 2r\\
                           &= 2r\cdot \sum_{l=0}^\infty \eta^l = \frac{2r}{1-\eta}
        \end{align*}
\end{proof}

Let $s\in\mathcal{S}$. Denote by $O(s)$ the number of output states that are visited from
$s$ before $\block{s}$ is left upward. Clearly $O(s)\leq N(s)$.

\begin{theorem}\label{thm:number_of_tests} Let $p_c< 1$ and $p_c\cdot r\cdot\max\set{p_1,\ldots,
    p_r}<1$. Then for any block $\alpha$
    \begin{displaymath}
        \expect{O(s_1^\alpha)}=\frac{\sum p_i}{1-p_c\sum p_i}\quad \text{ and } \quad
        \expect{N(s_1^\alpha)}=\frac{r + \sum p_i}{1-p_c\sum p_i}
    \end{displaymath}
\end{theorem}
\begin{proof}
    $C\defeq\expect{N(s_1^\alpha)}$ is finite by \cref{lem:exp_finite}. Note that $C$ is independent
    of $\alpha$ by the strong Markov property. We recall that $N(s_1^\alpha)=\HT[s_1^\alpha]{A}$ is
    a hitting time, where $A=\set{s_i^\beta\mid \beta \prefix \alpha}\cup\set{\bot}$. In the
    remainder of the proof, we drop the superscript Greek letter for all states in $\alpha$; i.e.
    $s_1$ is understood to mean $s_1^\alpha$.

    Every path from $s_1$ to $A$ must visit $s_2,\ldots, s_r$. Thus, by the strong Markov
    property, $N(s_1) = \left(\sum_{i=1}^{r-1}\HT [s_i]{s_{i+1}} \right)+ \HT[s_r]{A}$. By linearity of
    expectation and \cref{eqn:mean_hitting_times}
    \begin{align*}
        C =\expect{N(s_1)} &= \sum_{i=1}^{r-1}\MHT[s_i]{s_{i+1}} + \MHT[s_r]{A}
                            = \sum_{i=1}^{r-1} 1 + p_i\cdot \MHT[c_i]{s_{i+1}} + (1+p_r\cdot
                                        \MHT[c_r]{A})\\
                           & = \sum_{i=1}^{r-1} 1 + p_i(1 + p_c\cdot \underbrace{%
                                        \MHT[s_1^{\alpha\cdot i}]{s_{i+1}}}_{C})
                                        + (1+p_r(1 + p_c \cdot \underbrace{%
                                        \MHT[s_1^{\alpha\cdot r}]{A}}_{C}))\\
                           & = r + \sum_{i=1}^{r} p_i + Cp_c\sum_{i=1}^r p_i\tag{$\ast$}
    \end{align*}
    Solving for $C$ proves the second claim of the theorem.

    The proof for $\expect{O(s_1)}$ is similar. We first make a slight modification to
    \cref{eqn:mean_hitting_times} to count only output states:
\begin{equation*}
    \begin{array}{ll}
    \MHT[s]{A}=0  &\quad\text{if $s\in A$}\\
    \MHT[s]{A}=\begin{cases}
        1 + \sum_{s'\in\states} p(s,s')\MHT[s']{A}  & s\text{ is an output state}\\
        0 + \sum_{s'\in\states} p(s,s')\MHT[s']{A}  & s\text{ is not an output state}
    \end{cases}&\quad\text{if $s\notin A$}
\end{array}
\end{equation*}
This can be shown in exactly the same way as equations (1) (see e.g. \cite{norris1998markov} for the
proof of the classical theorem; the adaption is straightforward). Alternatively, the following
intuition can be turned into a formal proof:

Observe that in counting only output states $\mathcal{O}\subseteq\states$, we are effectively
studying a second Markov chain, whose state set consists only of output states (and $\bot$). For
$s, s' \in \states$ write $\mathsf{P}(s,s')$ for the set of all \emph{
simple paths} (without repeating vertices) from $s$ to $s'$ that do not visit $\mathcal{O}$. The
transition probabilities $p'$ of this second chain are given by the relation
\begin{displaymath}
    p'(s,s') = \sum_{w\in\mathsf{P}(s,s')} \prod_{i=1}^{|w|-1} p(w_{i-1},w_i)
\end{displaymath}
So our modified formulas are simply a different way to write
\cref{eqn:mean_hitting_times} for this modified chain in an iterative fashion.

With these modifications, we see that $(\ast)$ becomes:
\begin{displaymath}
C' = \sum_{i=1}^{r} p_i + C'p_c\sum_{i=1}^r p_i
\end{displaymath}
Again, solving for $C'$ establishes the claim.
\end{proof}


Note that for any given state $s_1^\alpha$, the mean hitting time
$\MHT[\sharp]{s_1^\alpha}=\sum_{n\geq 1}\Pr[\HT[\sharp]{s_1^\alpha}\geq n]\geq \sum_{n\geq 1}
1-p_c=\infty$ (where we use $\expect{X}=\sum_{n\geq 1} \Pr[X\geq n]$
for any random variable that only takes on positive integer values).
So although we have a non-zero probability of selecting every test, we won't, informally speaking,
do so on average. Naturally, this is solved by repeating the experiment a sufficient number of
times. This is the content of the next section.


\subsubsection{Infinite Looping and Time-To-Hit}\label{subsec:infinite}
As shown in \cref{lem:exp_finite}, the program in \cref{lst:generator_program} terminates
eventually. As a result, every state except $\bot$ in the Markov chain we studied above is transient
and, moreover, the number of produced test-cases is always finite. In testing, one aims at a high
test coverage, and the number of test cases we produce in this fashion, though free of duplicates,
has a low chance of visiting tests in deep blocks. A natural approach is to loop on the predicate
\texttt{t/1} like so:
\begin{lstlisting}[numbers=none]
main_loop(X) :- repeat, t(X).
\end{lstlisting}

With respect to our Markov chain this amounts to removing $\bot$ and to instead redirect any arc
into $\bot$ to $\sharp$. The resulting chain is recurrent (indeed positive recurrent) and  we compute
the mean hitting time of any state. In what follows, we will assume $p_1=p_2=\cdots=p_r\defeq p$
such that $r\cdot p \cdot p_c<1$ (as in \cref{thm:number_of_tests}).
Moreover, we assume that $p(\sharp,s_1^{\varepsilon})=1$, so that the empty list is never selected
as an output. This simplifies the formulas below slightly, but has otherwise no effect on the line
of reasoning we give here.


Given the conditions of \cref{thm:number_of_tests}, there is a constant $C=N(s_1^\alpha)$ that is
independent of the value of $\alpha$. As noted before, $C=\MHT[s_1^{\alpha}]{E_\alpha}$ is a mean
hitting time where $E_\alpha=\set{s_i^\beta\mid \beta \prefix\alpha,\;1\leq i\leq
r}\cup\set{\sharp}$ (note that we modified the definition of $E$ used in the previous section by
replacing $\bot$ by $\sharp$). Recall that we usually drop the subscript $\alpha$, because it is clear
from context.

If we hop from one state $s_i^\alpha$ to its neighbor $s_{i+1}^\alpha$ we might traverse the tree
below $s_{1}^{\alpha\cdot i}$ with probability $p\cdot p_c$. That step will visit $C$ states. This
means (by \cref{eqn:mean_hitting_times}):
\begin{displaymath}
    \MHT[s_i^\alpha]{s_{i+1}^\alpha} = 1 + p(1 + p_cC) \defeq \Delta
\end{displaymath}

More generally the mean hitting time within a block is again independent of $\alpha$ and can be computed
as:
\begin{equation}
    \MHT[s_1^{\alpha}]{s_i^{\alpha}} =
        \MHT[s_1^{\alpha}]{s_2^{\alpha}} + \MHT[s_2^{\alpha}]{s_3^{\alpha}}+\cdots
        + \MHT[s_{i-1}^\alpha]{s_i^{\alpha}}= (i-1)\Delta\label{eqn:hitting_within_block}
\end{equation}

We define the \emph{leave upward time} $U_{s_i^\alpha}=\MHT[s_i^\alpha]{E}$ where $E=E_\alpha$ as
above.  Note that the value $U_{s_i^\alpha}\in\naturals\cup \set{\infty}$ does not actually depend
on $\alpha$. This justifies writing $U_i=U_{s_i^\alpha}$. It is obvious that $C=U_{1}$. Moreover, by
using the same derivation as that in \cref{eqn:hitting_within_block}:
\begin{equation}
    U_i = (r-i + 1)\Delta \quad 1\leq i \leq r
    \label{eqn:leave_upward_computed}
\end{equation}

We already noted that $C=U_1$. A related quantity is the hitting time of $\sharp$ from
any $s_i^\alpha$, $\alpha=\alpha_1\cdots \alpha_t$, which we may compute using the intermediate leave upward times:
\begin{displaymath}
    \MHT[s_i^\alpha]{\sharp} = U_{i} + U_{\alpha_t + 1}+ \cdots
    U_{\alpha_1+1}
\end{displaymath}
Note that we abuse notation: \Cref{eqn:leave_upward_computed} gives $U_{r+1}=0$. While $s_{r+1}^\alpha$ does not
exist and hence the corresponding hitting time is not defined, it is convenient to allow such terms
and exploit that $U_{\alpha_j+1} = 0$ whenever $\alpha_j=r$ ($1\leq j\leq t$).

With this, we may compute:
\begin{align}
    \MHT[s_i^\alpha]{\sharp} &= U_i + \sum_{k=1}^t U_{\alpha_k+1} = \Delta\cdot \left((r-i+1) + \sum_{k=1}^t
                             (r-(\alpha_k+1) + 1) \right)\nonumber\\
                             &= \Delta\cdot \left(1+ \sum_{k=0}^t r- \alpha_k\right)\qquad \text{where }\alpha_0\defeq i\label{eqn:leave_upward}
\end{align}
Note again that the formula works correctly, if $i=r+1$: Say $\alpha = rrr$. Then we are in the
process of falling back to $\sharp$ and the equation gives 0. While the hitting time is again not
defined for the non-existent state $s_{r+1}$, we will sometimes have to compute the hitting time of
$\sharp$ from the \enquote{right neighbor} of $s_{i+1}$. In these situations, abusing notation in
this way is useful because we need not distinguish between cases where $i <r$ and those where $i=r$.

Finally, we may now compute the hitting time of an arbitrary state in terms of hitting times in
intermediate blocks, again using \cref{eqn:mean_hitting_times}. Let
$\alpha=\beta\cdot j$.

\begin{align*}
    &&\MHT[\sharp]{s_1^\alpha} = & \MHT[\sharp]{s_j^\beta}  + 1\\
    &&&+(1-p)(\MHT[s_{j+1}^\beta]{\sharp} + \MHT[\sharp]{s_i^\alpha}) & \text{(fall through to
$s_{j+1}^\beta$)}\\
    &&&+ p(1 + (1-p_c)(\MHT[s_{j+1}^\beta]{\sharp} + \MHT[\sharp]{s_i^\alpha})) & \text{(no visit
    to next block $\alpha$)}\\
    && = & \MHT[\sharp]{s_j^\beta}  + 1  + p+ (\MHT[{s_{j+1}^\beta}]{\sharp}+\MHT[\sharp]{s_i^\alpha})(1 - pp_c)
\end{align*}
This gives
\begin{displaymath}
\MHT[\sharp]{s_1^\alpha} = \frac{\MHT[\sharp]{s_j^\beta} + 1 + p +
(1-pp_c)\MHT[{s_{j+1}^\beta}]{\sharp}}{pp_c}
\end{displaymath}
and together with \cref{eqn:hitting_within_block} and \cref{eqn:leave_upward}, recalling that
$\alpha_{|\alpha|}=j$,  we have:
\begin{align}
    \MHT[\sharp]{s_i^\alpha} &= \frac{\MHT[\sharp]{s_j^\beta} + 1 + p +
    (1-pp_c)\MHT[{s_{j+1}^\beta}]{\sharp}}{pp_c} + (i-1)\Delta\nonumber\\
                            &= \frac{\MHT[\sharp]{s_j^\beta}}{pp_c}+\frac{1}{pp_c}\left(1 + p +
                            (1-pp_c)(\sum_{k=1}^{|\alpha|}r-\alpha_k)\Delta\right) +
                            (i-1)\Delta\label{eqn:recursive_hitting_time}
\end{align}

The following theorem gives a closed formula:
\begin{theorem}
    Let $s_i^\alpha$ for some $\alpha=\alpha_1\cdots \alpha_t$. Let
    $\nu=pp_c$ and $\nu\cdot r<1$. Then
\begin{equation}
    \MHT[\sharp]{s_i^\alpha} = \nu^{-t} + (i-1)\Delta + \sum_{k=1}^{t} \frac{1+ p + (\alpha_k-1)\Delta +
    (1-\nu)\sum_{s=1}^k(r-\alpha_s)\Delta}{\nu^{t+1-k}}\label{eqn:hitting_time_theorem}
\end{equation}
\end{theorem}
\begin{proof}
    By induction on $t$. If $t=0$, then $\alpha=\varepsilon$ and by \cref{eqn:hitting_within_block},
    we have $\MHT[\sharp]{s_i^\varepsilon}=1 + (i-1)\Delta$. Moreover the empty sum in
    \cref{eqn:hitting_time_theorem} equates to 0 establishing the induction base.

    Now let $t>0$ and assume the statement holds for $t-1$. By induction, we may replace $\MHT[\sharp]{s_j^\beta}$
in \cref{eqn:recursive_hitting_time} with \cref{eqn:hitting_time_theorem}:
    \begin{align*}
             &\frac{1}{\nu}\left(\nu^{-(t-1)}+ (\alpha_t-1)\Delta +
                \sum_{k=1}^{t-1}
            \frac{1+p+(\alpha_k-1)\Delta+(1-\nu)\sum_{s=1}^k(r-\alpha_s)\Delta)}{\nu^{t-k}}\right)\\
                                      &+ \frac{1}{\nu}\left(1 + p +
                                      (1-\nu)(\sum_{s=1}^{t}r-\alpha_s)\Delta\right) + (i-1)\Delta\\
                = &\nu^{-t}+ \sum_{k=1}^{t-1}
                \frac{1 + p +
                (\alpha_k-1)\Delta+(1-\nu)\sum_{s=1}^k(r-\alpha_s)\Delta)}{\nu^{t+1-k}}\\
                  &+ \frac{(1 + p + (\alpha_t-1)\Delta +
                  (1-\nu)(\sum_{s=1}^{t}r-\alpha_s)\Delta)}{\nu^{t+1 - t}} +
                  (i-1)\Delta\\
                = &   \nu^{-t} + (i-1)\Delta + \sum_{k=1}^{t} \frac{1+ p + (\alpha_k-1)\Delta +
                (1-\nu)\sum_{s=1}^k(r-\alpha_s)\Delta}{\nu^{t+1-k}}
                \end{align*}
            \end{proof}
\begin{corollary}
    \label{cor:hitting_time_guard_theta}
    Let $\nu=p\cdot p_c$ with $\nu\cdot r<1$. Then $\MHT[\sharp]{s_i^\alpha}\in\bigTheta(\nu^{-t})$
    for any $\alpha=\alpha_1\cdots\alpha_t$.
\end{corollary}
\begin{proof}
    Write \cref{eqn:hitting_time_theorem} as
    \begin{displaymath}
    \nu^{-t} + A + \nu^{-t-1}\sum_{k=1}^t \frac{B_k +\sum_{s=1}^kD_{s}}{\nu^{-k}}
    \end{displaymath}
        for suitable constants (in $t$) $A\geq 0$, $B_k\geq 0$, and
        $D_{s}\geq 0$, whereby $\MHT[\sharp]{s_i^\alpha}\in\bigOmega(\nu^{-t})$.

    Choose suitable largest values $B\geq B_k$ for all
    $t\in\naturals$, $1\leq k\leq t$, and $D\geq D_{s}$ for all $1\leq
    s\leq t$. Bound \cref{eqn:hitting_time_theorem} from above by
    \begin{displaymath}
        \nu^{-t} + A + \nu^{-t-1}\sum_{k=1}^t \frac{B + D\cdot k}{\nu^{-k}}
        \leq \nu^{-t} + A + \nu^{-t-1}(B+D)\cdot\sum_{k=1}^{t} \frac{k}{\nu^{-k}}
    \end{displaymath}
    A well-known calculation via derivatives gives $\sum_{k=1}^{t}k\cdot
    \nu^k=\nu\sum_{k=1}^{t}k\nu^{k-1} \leq \nu\sum_{k=1}^{\infty}k\nu^{k-1} =
        \nu\cdot\differential{\nu}\sum_{k=0}^{\infty}\nu^k =\frac{\nu}{(1-\nu)^2}$.
    With that we have
    \begin{displaymath}
        \nu^{-t} + A + \nu^{-t-1}\sum_{k=1}^t \frac{B + D\cdot k}{\nu^{-k}}\leq
        \nu^{-t}+A+(B+D)\frac{\nu^{-t}}{(1-\nu)^2}\in\bigO(\nu^{-t})
    \end{displaymath}
\end{proof}

\subsection{Drop-and-Shuffle Strategy}\label{sec:ds}
\subsubsection{Number of Generated Tests}

To study the number of generated tests in this context, we again have to define a Markov chain. The
Markov chain in the shuffle and drop scenario is significantly more complicated than the one we
studied previously in \cref{subsec:guard}. This is because there are now multiple ways a given
test-case can be output. To distinguish between those, we need to use a larger and more complex
state set. We will illustrate this, before defining the Markov chain formally.

Consider again the program in \cref{lst:generator_program}. When the current goal is
\texttt{command(H)}, we have a set $R$ of rules whose head unifies with this goal. In this case,
$R=\set{\texttt{command(H) :- command1(H)},\ldots,\texttt{command(H) :- commandr(H)}}$. Recall this
is meant to represent $r\in\naturals$ distinct rules. In standard SLD-resolution, we would select
the first rule that occurs in the input program $\mathcal{P}$, namely \texttt{command(H) :-
command1(H)} first, and push the remaining $r-1$ rules on the stack from right to left. During
backtracking, we would then eventually explore each of those rules in the order given in the input
program (except in case of an infinite recursion).

In the drop-and-shuffle strategy, we first perform an independent Bernoulli trial for each rule
$\rho\in R$: With probability $p_d$, we remove $\rho$ from $R$. We refer to $p_d$ as the \emph{drop
probability}. In this way, a set $R'\subseteq R$ is computed. The random variable $|R'|$ follows a
Binomial distribution: $\Pr[|R'|=k] = \binom{|R|}{k}p_d^{|R|-k}(1-p_d)^k$.
Next, we shuffle the set $R'$. To this end, we select a permutation $\pi\in \mathbb{S}(R')$, where
$\mathbb{S}(M)$ denotes the symmetric group on a given set $M$. Conceptually, any probability
distribution on $\mathbb{S}(R')$ is conceivable. In this paper, we follow a simpler approach and
select $\pi$ uniformly at random from $\mathbb{S}(R')$. In this way, we obtain an ordered tuple of
elements of $R$ without any repetitions.

The result of these two random processes is a tuple $(\rho_{i_1}, \ldots, \rho_{i_k})$ where $k=|R'|$ and
$\rho_{i_j}\in R$. To simplify notation, we identify $R'$ with this tuple in what follows. Since both
random events -- dropping and shuffling -- are independent, the probability of each such tuple
$R'=(\rho_{i_1},\ldots, \rho_{i_k})$ is precisely $\Pr[R'] = \frac{p_d^{n-k}(1-p_d)^k}{k!}$.

\begin{remark}
    Note that this random process is different from the classical \enquote{drawing
    without replacement}, where the number of elements that are drawn is usually a fixed parameter.
\end{remark}

These observations motivate the following Markov chain $\mkvDS=(\statesDS,(X_n)_{n\in\naturals_0},
p, \varepsilon)$. The state set $\statesDS$ now consists of stacks of choice-points --- in loosely
the same way a Prolog runtime would maintain them. Before formally defining the state set and
transition probabilities, we invite the reader to consider a simplified graphical representation of
the chain, as given in \cref{fig:chain_shuffle_drop}.

\Cref{fig:chain_shuffle_drop} gives an overview of the chain. Some details have been omitted or
simplified to avoid cluttering the picture. Probabilities are not shown. We have omitted choice
points from a higher layer: A state/stack of the form $[H|T][](1,5,2,3)[H|T](1,2,3)$ is thus simply
represented as $(1,2,3)$, omitting the \enquote{lower} parts of the stack. Note that these items are
implicitly clear from the path to a given node. Moreover, most backtracking arrows have been
omitted. Finally, at any given depth, both the node $[H|T][]$ and the node $[H|T]$ each have
$\sum_{k=0}^r\binom{r}{k}\cdot k!=\floor{r!\cdot \mathsf{e}}$ children\footnote{We have
    $\sum_{k=0}^r\binom{r}{k}k!=r!\sum_{k=0}^{r}\frac{1}{k!}\approx r!\cdot \mathsf{e}$ with error
    term $\sum_{k=r+1}^{\infty}\frac{1}{k!}\in (0,1)$, because $r\geq 1$.}. These have also mostly
been omitted.

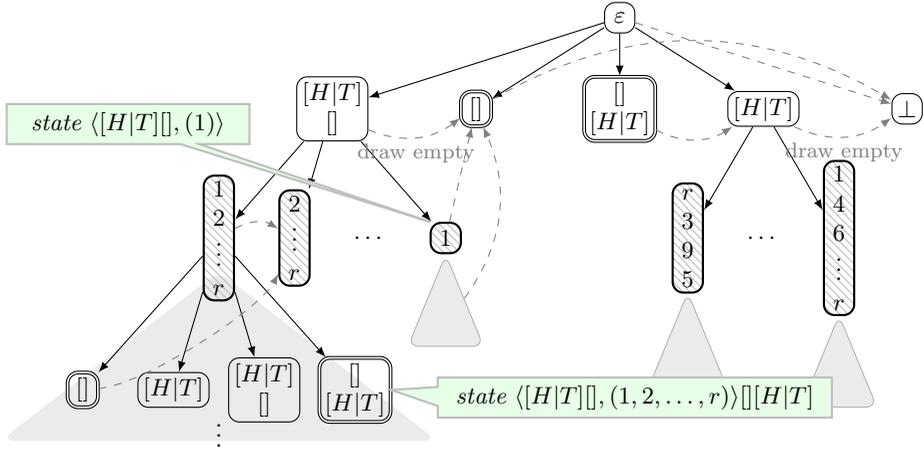
\begin{figure}[t]
    \centering
\begin{tikzpicture}[every node/.style={draw,rounded corners, rectangle, minimum size=.4cm,
    align=center,font=\small,inner sep=2pt},grow=down,edge from parent/.style={draw,-latex},
    level 1/.style={level distance=1.2cm, sibling distance=1.9cm},
    level 2/.style={level distance=1.7cm, sibling distance=1cm},
    level 3/.style={level distance=2cm},
    level 4/.style={level distance=.8cm},
    ]
    \tikzstyle{shallowfork}=[edge from parent path={(\tikzparentnode\tikzparentanchor) -- +(0,-1.2cm) -|
    (\tikzchildnode\tikzchildanchor)}]
    \tikzstyle{recroot}=[thick,draw,pattern=north west lines,pattern color=lightgray]
    \tikzstyle{labelstyle}=[draw=gray!50!white,sharp corners,fill=none,thick,fill=green!10!white,
    font={\small\itshape},inner sep=3pt,shape=rectangle callout]

    \node[shape=isosceles triangle, rounded corners, thick, lightgray, fill,
        opacity=.3,rotate=90,minimum width = 5.7cm, minimum height=2.2cm, isosceles triangle stretches]
        at (-5.3,-4.6) {};
    \node at (0,0) (sharp) {$\varepsilon$}
        child { node (sel1) {$[H|T]$\\$[]$}
            child {
                node[recroot] (c1tor) {$1$\\$2$\\$\vdots$\\$r$} [sibling distance=1.2cm]
                child[] {node (leaf1) [double]{$[]$}
                }
                child[] {node{$[H|T]$}
                }
                child[] {node{$[H|T]$\\$[]$}
                }
                child[] {node[double] (labeledleaf) {$[]$\\$[H|T]$}
                }
            }
            child {node[recroot] (c2tor) {$2$\\$\vdots$\\$r$}}
            child {node[draw=none] {$\cdots$} edge from parent[draw=none]}
            child {node[recroot] (singletonchild) {$1$}
                    child [level distance=.2cm]{ node[anchor=north,isosceles triangle,shape border
                    rotate=45,fill=lightgray,opacity=.3, minimum size=1cm] (recurtree2) {} edge from parent[draw=none]}
            }
        }
        child {node[double] (singoutput) {$[]$} }
        child {node[double] (l12) {$[]$\\$[H|T]$}}
        child {node (l13) {$[H|T]$}
            child {node[recroot] (rightchild1) {$r$\\$3$\\$9$\\$5$} child [level distance=.7cm]{ node[anchor=north,isosceles triangle,shape border
                rotate=45,fill=lightgray, opacity=.3,minimum size=1cm] {} edge from parent[draw=none]}
            }
            child {node[draw=none] {$\cdots$} edge from parent[draw=none]}
            child {node[recroot] {$1$\\$4$\\6\\$\vdots$\\$r$}
                child [level distance=1cm]{ node[anchor=north,isosceles triangle,shape border
                rotate=45,fill=lightgray, opacity=.3,minimum size=1cm] {} edge from parent[draw=none]}
            }
        }
        child {node(bot) {$\bot$} edge from parent[dashed, gray]}
        ;

        \node[draw=none] at (-5.3,-5.4) {$\vdots$};

    \path[-latex, dashed,gray] (leaf1) edge[bend right,looseness=0.5] (c2tor)
        (recurtree2) edge[bend right] (singoutput)
        (singletonchild) edge (singoutput)
        (singoutput)    edge[bend left] (bot)
        (l12) edge[bend right] (l13)
        (sel1) edge[bend right] node[draw=none,below] {\footnotesize draw empty} (singoutput)
        (l13) edge[bend right] node[draw=none,below] {\footnotesize draw empty} (bot)
        (c1tor) edge[bend left] (c2tor)
    ;

\node[labelstyle,text width=3cm,callout absolute pointer={(singletonchild.north west)}] (label1)
        at (-6.5,-1.4) {state $\state{[H|T][],(1)}$ };
\node[labelstyle,text width=5cm,callout absolute pointer={(labeledleaf.east)}] (label2)
        at (0.2,-5) {state $\state{[H|T][],(1,2,\ldots,r)}[][H|T]$ };
\end{tikzpicture}

\caption{The Markov chain corresponding to $\prog{P}$ in the Drop-and-Shuffle approach. Dashed
    arrows represent backtracking. Double circled nodes produce an output. Hatched nodes are
    \emph{recursive sub-tree roots}, which start an entire infinite
sub-tree with the same structure as the whole chain (shown as gray triangles).} \label{fig:chain_shuffle_drop}
\end{figure}

We describe the state set via (the language
defined by) a regular expression. First we need two auxiliary languages: \begin{align*}
    \states_{\mathsf{Sel}}=& \set{[], [H|T], [][H|T], [H|T][]}\\
    \states_{\mathsf{Com}} = & \set{(y_1,\ldots, y_l)\mid 1\leq l \leq r,\;y_i\in \set{1,\ldots, r}, y_i\neq y_j \text{ for all }
    i\neq j}
\end{align*}

The set $\states_{\mathsf{Com}}$ corresponds to all ordered subsets of \texttt{command} rules,
i.e. subsets of $\set{1,\ldots, r}$. Note that the empty subset is excluded, i.e.\
$()\notin\states_{\mathsf{Com}}$. The set $\states_{\mathsf{Sel}}$ corresponds to the five
probabilistic options the drop-and-shuffle algorithm gives us for resolving goals of the form
\texttt{t(X)}, including potential choice-points for backtracking. Note that, again, the
\enquote{empty selection} $()\notin\states_{\mathsf{Sel}}$ is not included.

We can now define the state set $\statesDS=(\states_{\mathsf{Sel}}\times \states_{\mathsf{Com}})^*\cdot
(\states_{\mathsf{Sel}}+\varepsilon)$. The
\enquote{empty state} $\varepsilon\in\statesDS$ is the initial state of the chain; i.e.
$\Init=\varepsilon$. We write the pairs $\state{x,y}\in\states_{\mathsf{Sel}}\times
\states_{\mathsf{Com}}$ in angular brackets in order to visually distinguish them and improve
readability.

Given a state $\varepsilon\neq s\in\statesDS$, we may write it as $s=w\cdot x$ or $s=w\cdot
\state{x,y}$ with $w\in\statesDS$, $x\in \states_{\mathsf{Sel}}$, and $y\in \states_{\mathsf{Com}}$. Because
$\varepsilon\notin \states_{\mathsf{Sel}}$ and also $\varepsilon\notin\states_{\mathsf{Com}}$, this
factorization is unique. We make liberal use of this observation when defining the transition
probabilities. We first define transitions that \emph{descend} further into the tree. These
correspond to solid arrows in \cref{fig:chain_shuffle_drop}. For any $l\geq 1$, $w\in \statesDS$,
and $x\in\states_{\mathsf{Sel}}$:
\begin{align*}
    p(w,\,w\cdot x)
    =&\frac{(1-p_d)^2}{2} & x\in\set{[H|T][],[][H|T]}\\
    p(w,\,w\cdot x) =& (1-p_d)p_d & x\in\set{[],[H|T]}\\
    p(w\cdot x,\, w\cdot \state{x, (y_1,\ldots, y_l)}) =& \frac{p_d^{r-l}(1-p_d)^{l}}{l!} & x \in
    \set{[H|T][], [H|T]}
\end{align*}

Next, we define transitions that correspond to \emph{backtracking}. These correspond to dashed
arrows in \cref{fig:chain_shuffle_drop}. To this end we define the operation
$\mapping{\shorten}{\statesDS}{\statesDS}$. Intuitively, this operation pops from the stack until we
arrive at a previously unpursued choice point. Looking back at \cref{fig:chain_shuffle_drop}, it
identifies the target of the backtracking arrow. It  may be necessary to remove multiple layers of
pairs $\state{x,y}$ when backtracking. For example  $\shorten(\state{[H|T][],(1,3)}\state{
[H|T][],(3)}\state{[H|T],(1)})=\state{[H|T][],(1,3)}[]$. Formally, we define $\shorten$ recursively
with $\shorten(\varepsilon)=\bot$ and:
\begin{align*}
    &\shorten(w\cdot [][H|T])=w\cdot [H|T] & &\shorten(w\cdot [H|T][])=w\cdot []\\
    &\shorten(w\cdot [])=\shorten(w) & &    \shorten(w\cdot[H|T])=\shorten(w)\\
    &\shorten(w\cdot \state{x,(y_1)})=\shorten(w\cdot x)  & &  \shorten(w\cdot \state{x,(y_1,y_2,\ldots,y_l)})= w\cdot\state{x,(y_2,\ldots,y_l)}
\end{align*}
We can now define all backtracking transitions as follows:
\begin{align*}
    p(w\cdot x, \shorten(w\cdot x)) =& 1  & x = [] \text{ or }  x= [][H|T]\\
    p(w\cdot x, \shorten(w\cdot x)) =& p_d^r & x=[H|T] \text{ or } x=[H|T][]\\
    p(w, \shorten(w)) =& p_d^2 & w=w'\cdot\state{x,y} \text{ or } w=\varepsilon
\end{align*}
The last two model the event that all matching rules are dropped when unifying \texttt{command(X)}
or \texttt{t(X)}, respectively.

Note the recursive structure of $\mkvDS$: Any state of the form $w\cdot\state{x,y}$
($x\in\states_{\mathsf{Sel}},\;y\in\states_{\mathsf{Com}}$) is the root of an infinite sub-tree that
has a structure identical to that of $\mkvDS$. We call such states \emph{recursive sub-tree roots}
or simply \emph{sub-tree roots}. These states are drawn with hatched background in
\cref{fig:chain_shuffle_drop}.

To each recursive sub-tree root $s$ corresponds a unique state $\Exit_s$ that is visited when the
chain leaves that sub-tree (via backtracking). In other words: All paths
that exit the sub-tree below $s$ must traverse $\Exit_s$. For example, consider the left-most gray
sub-tree in \cref{fig:chain} (with the large gray triangle in the background) below
$s=\state{[H|T][],(1,2,\ldots,r)}$. Its exit-state is the next node to the right:
$\Exit_s=\state{[H|T][],(2,\ldots,r)}$. In general, if $s$ is of the form $w\state{x,y}$, then
$\Exit_s=\shorten(w\cdot\state{x,y})$. Note that $\Exit_s$ may be at the same depth as $s$ or at a
lower depth than $s$. It is never at a higher depth.

We can now compute the average number of generated tests as before. It is again the mean hitting time
$\MHT[\varepsilon]{\bot}$. By an argument identical to \cref{prop:exit_prob1}, the chain will reach
$\bot$ eventually with probability 1, if $p_d > 0$. But that does \emph{not} imply that the hitting
time --- an expected value --- converges for all such values of $p_d$. Indeed, the hitting time is
finite only for a subset of possible choices for $p_d>0$, as the following theorem shows:
\begin{theorem}\label{thm:number_of_tests_ds}
    Let $p_d\in (1-\frac{1}{\sqrt{r}},1]$. Then the expected number
    $\MHT[\varepsilon]{\bot}$ of states visited from $\varepsilon$ is finite and given by:
    \begin{displaymath}
       \MHT[\varepsilon]{\bot} = \frac{1 + 2(1-p_d)}{1- r(1-p_d)^2}
    \end{displaymath}

Moreover this hitting time is identical to $\MHT[s]{\Exit_s}$ for any recursive sub-tree root $s$
and its corresponding exit state $\Exit_s$. We define $C\defeq\MHT[\varepsilon]{\bot}$ and remark
that it is a constant property of the chain.
\end{theorem}
\begin{proof}
The fact that $\MHT[\varepsilon]{\bot}=\MHT[s]{\Exit_s}$ for any recursive sub-tree root $s$ is
apparent from the definition of the chain: The transition probabilities are prefix invariant. So
$p(x\cdot a, x\cdot b) = p(a,b)$ for all factorisations $s=xa$ with $x\in\statesDS$.
$\MHT[\varepsilon]{\bot}$ is the unique minimal positive solution to the equations
\cref{eqn:mean_hitting_times}. Now, letting $C\defeq\MHT[\varepsilon]{\bot}$:
\begin{align*}
    C =& 1 + p_d(1-p_d)\cdot 1 + p_d^2\cdot 0\\
      & + p_d(1-p_d)\cdot \left(1+ \sum_{k=0}^r \frac{p_d^{r-k}(1-p_d)^k}{k!}\cdot k!\cdot
\binom{r}{k}\cdot k\cdot C\right) \\
      & + \frac{1}{2}\cdot(1-p_d)^2\cdot \left(2 + \sum_{k=0}^r \frac{p_d^{r-k}(1-p_d)^k}{k!}\cdot k!\cdot
\binom{r}{k}\cdot k\cdot C\right)\\
      & + \frac{1}{2}\cdot(1-p_d)^2\cdot \left(1 + \sum_{k=0}^r \frac{p_d^{r-k}(1-p_d)^k}{k!}\cdot k!\cdot
\binom{r}{k}\cdot (k\cdot C + 1)\right) 
\end{align*}

We recall that $(x+y)^r=\sum_{k=0}^{r}\binom{r}{k}x^k y^{r-k}$. Differentiation and subsequent
multiplication by $x$ gives $rx(x+y)^{r-1} = \sum_{k=0}^rk\binom{r}{k}x^{k}y^{r-k}$. Moreover,
recall that $\sum_{k=0}^{r}\binom{r}{k}p_d^{r-k}(1-p_d)^k=1$. With this, the above simplifies to
\begin{displaymath}
    C = 1 + 2(1-p_d) + Cr(1-p_d)^2
\end{displaymath}
Now if $p_d\leq 1- \frac{1}{\sqrt{r}}$, then this implies $C \geq 1 + \frac{2}{\sqrt{r}} + C$, which
is possible only if $C=\infty$. On the other hand, if $p_d\in (1-\frac{1}{\sqrt{r}},1]$ then solving
for $C$ establishes the claim. Note that on this interval, the formula has no singularities and is
positive.
\end{proof}

\subsubsection{Infinite Looping and Time-To-Hit}\label{subsec:ds:infinite}

We again construct a recursive analysis of the mean-hitting-time. Recall that in
\cref{subsec:infinite} we analyzed the hitting time of the unique state corresponding to
$\tau=(\tau_1,\ldots,\tau_l)$ (with $\tau_i\in\set{1,\ldots,r}$) by first considering the hitting
time of the unique state corresponding to $\tau^{(l-1)}=(\tau_1,\ldots, \tau_{l-1})$, and then
constructing the full hitting time from that number \enquote{bottom-up}. This recursive argument works less well in the
present scenario, where many different states correspond to $\tau^{(l-1)}$. Computing the overall
hitting time as a sum of those intermediate hitting times is challenging, as we explain
below. We therefore develop an approach to compute the hitting time \enquote{top-down}.

First, we need to add looping to the Markov chain, as in \cref{subsec:infinite}. Recall that there
we merged the two states $\sharp$ and $\bot$. We do \emph{not} do that here. Instead, we add an edge
from $\bot$ to $\varepsilon$ with probability 1. This is for technical reasons, and we will justify
this choice further below.

\begin{figure}[t]
    \begin{center}
    \begin{tikzpicture}[every node/.style={draw,rounded corners, rectangle, minimum size=.4cm,
        align=center,font=\small,inner sep=2pt},
        grow via three points={one child at (2.2,-.1) and two children at (2.2,-.1) and (1.4,-1.5)},
        edge from parent/.style={draw,-latex},transform shape]

        \node (root) at (0,0) {$[H|T]$}
            child { node[fill=green!30!white] (i1) {{1}\\{2}\\$\vdots$\\{$r$}}
                child { node (i2) {{2}\\$\vdots$\\{$r$}}
                    child { node[draw=none] (dots1)  {$\cdots$}
                            child { node (i3)  {{$r$}} }
                    }
                }
            }
            child { node[draw=none,rotate=-50] {$\vdots$} edge from parent[draw=none]}
            child { node (i11) {{2}\\{4}\\{1}\\$\vdots$\\{5}}
                child { node (i12) {{4}\\{1}\\$\vdots$\\{5}}
                    child { node[fill=green!30!white] (i13) {{1}\\$\vdots$\\{5}}
                        child { node[draw=none] (dots11) {$\cdots$}
                            child { node (i14) {{5}}}
                        }
                    }
                }
           };
       \node[anchor=north,isosceles triangle,shape border
                        rotate=45,fill=green,opacity=.3, minimum size=1cm] at (2.2,-1)  {} ;
       \node[anchor=north,isosceles triangle,shape border
                        rotate=45,fill=lightgray,opacity=.3, minimum size=1cm] at (4.4,-1)  {} ;
       \node[anchor=north,isosceles triangle,shape border
                        rotate=45,fill=lightgray,opacity=.3, minimum size=1cm] at (8.8,-1)  {} ;

       \node[anchor=north,isosceles triangle,shape border
                        rotate=45,fill=lightgray,opacity=.3, minimum size=1cm] at (0.6,-3.8)  {} ;
       \node[anchor=north,isosceles triangle,shape border
                        rotate=45,fill=lightgray,opacity=.3, minimum size=1cm] at (2.8,-3.8)  {} ;
       \node[anchor=north,isosceles triangle,shape border
                        rotate=45,fill=green,opacity=.3, minimum size=1cm] at (5,-3.8)  {} ;
       \node[anchor=north,isosceles triangle,shape border
                        rotate=45,fill=lightgray,opacity=.3, minimum size=1cm] at (9.4,-3.8)  {} ;

\end{tikzpicture}
\end{center}
\caption{Recursive sub-tree roots are connected to \emph{all} ordered subsets of
    $\set{1,\ldots,r}$ (most arrows omitted for readability). Each 
    subset of size $k$ starts \emph{chain} of $k$ subsets. One element (here $1$) is the desired
    next item of the test-sequence, and each chain contains at most one state with the desired
item at the top (depicted in green).} \label{fig:stack_chain}
\end{figure}
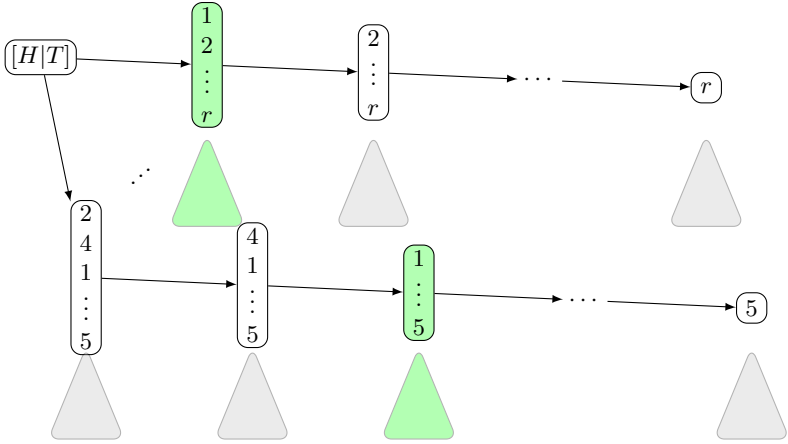

When running the program $\mathcal{P}$ in the drop-and-shuffle strategy, first we randomly select
and permute a subset of the two rule-heads \texttt{t([])} and \texttt{t([H|T])}. To proceed, we need
to visit $[H|T]$ or $[H|T][]$ at depth 0. Once that happens, there are two options: With probability
$p_d$, the next state does \emph{not} contain $\tau_{1}$, and thus we \emph{cannot} visit $\tau$ in
this loop iteration (i.e. before returning to $\varepsilon$ first). We call such a sub-tree root
\emph{unproductive}. Conversely, with probability $(1-p_d)$, the next state does contain $\tau_{1}$
somewhere within its stack (though not necessarily at the top). Those sub-tree roots and their
corresponding sub-trees are called \emph{productive} (shown in green in \cref{fig:stack_chain}). If
a productive or unproductive state does not have $\tau_{1}$ at its top, there is an infinite tree
below it that will we be explored, but cannot yield the desired test-case. We also call such a
sub-tree \emph{unproductive} (shown in gray in \cref{fig:stack_chain}). Note that productive sets of
size $k>0$ give rise to precisely $k-1$ unproductive sub-trees, though not all of those are
traversed \emph{before} the sub-tree containing $\tau$ is visited. Let $A_\tau$ denote the set of
states that output $\tau$:
\begin{fact}\label{fact:return_home}
If the chain arrives at an unproductive state, it must first visit $\varepsilon$ before visiting
$A_\tau$.
\end{fact}
Note that we are only talking about depth 0 for now, so this fact is obvious. At higher depths, we
would need to adjust the definition of \enquote{productive} to ensure that all prefixes at lower
depths have $\tau_1,\tau_2,\ldots$ at the top.

If the chain arrives at a productive state, the situation is more complex: The chain needs to
traverse a number of unproductive sub-trees, depending on the position of $\tau_{1}$ in the ordered
set. For example, in the second branch shown in \cref{fig:stack_chain}, there are two unproductive
sub-trees to be traversed \emph{before} reaching productive sub-tree (and several, indicated by the
dots, after). Below the productive state, we find a tree of recursive structure: The  goal now is to
reach a test-case of length $l-1$, or to return to $\varepsilon$ and start from scratch (cf.\
\cref{fact:return_home}).

Starting from a productive
state $s$ of size $m\leq r$ with item $\tau_{1}$ at position $k\in\set{0,\ldots, m}$, the hitting
time is thus $kC+\MHT[x]{A_\tau}$, where $C$ is the quantity from \cref{thm:number_of_tests_ds} and
$x\in\states$ is the state arising from $s$ after $k$ items have been popped (e.g. the green state
in the lower branch in \cref{fig:stack_chain}).

We stress that $\MHT[x]{A_\tau}$ depends not only on the remaining test-case suffix
$(\tau_2,\ldots, \tau_l)$, but also on the number of (unproductive) backtracking steps --- at least
$(m-k-1)C$ --- before returning to $\varepsilon$. The chain may arrive at an unproductive step at
some point, and by \cref{fact:return_home}, it first needs to visit $\varepsilon$ before reaching
$A_\tau$. Note that this may happen at depth $>1$ as well! This prevents us from applying induction
in a straightforward way. We thus need the following lemma, which allows us to reduce the hitting
time $\MHT[x]{A_\tau}$ to three quantities we may compute individually and inductively:

\begin{lemma}\label{lem:reduce_MHT} For any recurrent Markov chain, let $A\subseteq\states$ and
    $x,y\in\states\setminus A$ with $x\neq y$. Then:
    \begin{displaymath}
        \MHT[x]{A}=\MHT[x]{A\cup\set{y}} + \Pr[\HT[x]{y} < \HT[x]{A}]\cdot \MHT[y]{A}
    \end{displaymath}
\end{lemma}
We prove the lemma at the end of this section, as it is purely Markov theoretic. In our immediate
setting, the consequence is that

\begin{align}
&\MHT[\varepsilon]{A_{\tau}}=\MHT[\varepsilon]{A_{\tau}\cup\set{\bot}} + q\cdot \MHT[\bot]{A_{\tau}}
=\MHT[\varepsilon]{A_{\tau}\cup\set{\bot}} + q(1 +
\MHT[\varepsilon]{A_{\tau}})\nonumber\\
\iff & \MHT[\varepsilon]{A_{\tau}}=\frac{\MHT[\varepsilon]{A_{\tau}\cup\set{\bot}} + q }{1-
q}\label{eqn:iterations_ds_full}
\end{align}
for some
probability $q$ that depends on $\tau$. We shall see below that $q$ in fact only depends on
the length $l$ of the test-case. Note that by not merging $\bot$ and $\varepsilon$, we fulfill the
\enquote{$x\neq y$ premise} of the lemma. By choosing $p_d\in(1-\frac{1}{\sqrt{r}},1)$, we obtain a recurrent chain.

It is thus sufficient to compute $q$ and $\MHT[\varepsilon]{A_\tau\cup\set{\bot}}$. We have
\enquote{widened} the set $A_\tau$ by including $\bot$, which means we no longer have to worry about
looping back to $\varepsilon$. This effectively allows us to compute the desired hitting-time
inductively, applying \cref{lem:reduce_MHT} iteratively. We first turn to computing $q$ as a
function of $\tau$.

Let $s=\state{x_1,y_1}\cdots\state{x_m,y_m}\in\statesDS$ be any recursive sub-tree root, where
$y_i=(y_{i,1},\ldots, y_{i,n_i})\in\states_{\mathsf{Com}}$ for $1\leq i \leq m$. Then the sequence
$y_{1,1}\cdots y_{m,1}$ is the \emph{prefix} of $s$. It corresponds to the sequence that would be
output if state $s\cdot[]$ is reached.
\begin{proposition} \label{prop:recursion_ds}
    \begin{enumerate}
        \item \label{prop:recursion_ds:1} Let $s$ be a recursive sub-tree root with prefix $\rho$ and $\tau=(\tau_1,\ldots, \tau_l)$. Then
            \begin{displaymath}
                \Pr[\HT[s]{A_{\rho\tau}}>\HT[s]{\Exit_s}] =  1 - (1-p_d)^{2l+1}\defeq q^{(l)}
           \end{displaymath}
           In particular, this probability depends only on $l$ and not on the values $\tau_i$ or
           $\rho$.
       \item \label{prop:recursion_ds:2} Let $s$ be a recursive sub-tree root with prefix $\rho$ and let $\tau$ be any
            test-case. Then:
            \begin{displaymath}
                \MHT[s]{A_{\rho \tau}\cup\set{\Exit_s}}=\MHT[\varepsilon]{A_{\tau}\cup
                \set{\bot}}
            \end{displaymath}
            Informally, the stated mean hitting time is \enquote{translation invariant}.
    \end{enumerate}
\end{proposition}
\begin{proof}
\begin{enumerate}
    \item
    If $l=0$, then $\Pr[\HT[s]{A_\rho} > \HT[s]{\Exit_s}] = (1-p_d)p_d + p_d^2 =p_d$, which is the probability of not selecting $[]$,
    $[H|T][]$ or $[][H|T]$, each of which would eventually output the prefix $\rho$. This depends
    only on $l=0$, so write $q^{(0)}=p_d$.

    Now let $l>0$: We reach $\Exit_s$ without visiting the next deeper state by backtracking to it
    immediately with probability $p_d^2$, or by visiting $[]$ with probability $(1-p_d)p_d$ which
    sums up to $p_d$. Otherwise we reach the next deeper state with probability $(1-p_d)$. Here we either
    select $\tau_1$ and apply induction using \cref{fact:return_home}, or we do not select $\tau_1$.
    This gives:
    \begin{displaymath}
        \Pr[\HT[s]{A_{\rho\tau}}>\HT[s]{\Exit_s}] = p_d + (1-p_d)((1-p_d)\cdot q^{(l-1)} + p_d )
    \end{displaymath}
    which is of the form $A\cdot q^{(l-1)} + B$ with $A=(1-p_d)^2$ and $B=p_d(2-p_d)$. This gives
    rise to a polynomial (in $A$), namely: $q^{(0)}\cdot A^l + B\cdot\sum_{k=0}^{l-1}A^k$. Computing
    the geometric sum and simplifying the term proves the first claim.
\item This follows immediately from the definition of $\Exit_s$ and the recursive structure of the
    chain.
\end{enumerate}
\end{proof}

In what follows, we write $\tau^{(i)} = (\tau_{1},\ldots,\tau_i)$. So $\tau=\tau_l$ and $\tau_0=()$
is the empty test sequence. Write $A_i$ for the set of states that output $\tau^{(i)}$. For
simplicity, we write $h_i=\MHT[\varepsilon]{A_i\cup \set{\bot}}$.
\begin{lemma}
    Let $p_d\in (1-\frac{1}{\sqrt{r}},1)$ and $i\in \set{0,\ldots, l}$ and let $C$ denote the constant from
\cref{thm:number_of_tests_ds}. Then
\begin{multline*}
    h_{i} = (1-p_d)^{2i}\bigg(1+\frac{(1-p_d^2)(1+C)}{2} -\frac{C(1-p_d)^2(r-1)}{p_d(2-p_d)}\\
        - \frac{(1-p)(1+p) + (1-p)^3}{2p_d(2-p_d)} - i\cdot (1-p_d)^{2}\cdot \frac{1+C(r-1)}{2}\bigg) \\
        + \left(\frac{C(1-p_d)^2(r-1)}{p_d(2-p_d)} + \frac{(1-p)(1+p) +
    (1-p)^3}{2p_d(2-p_d)}\right)
\end{multline*}
\end{lemma}
\begin{proof}
For $h_0$ we need the hitting time of the set of four states $\bot$, $[]$, $[H|T][]$, and $[][H|T]$.
By \cref{eqn:mean_hitting_times} we have
\begin{displaymath}
    h_0=1 + (1-p_d)p_d(1+C) + \frac{1}{2}(1-p_d)^2(1+C)= 1+\frac{1}{2}(1-p_d^2)(1+C)
\end{displaymath}

To compute $h_{i+1}$, we first look at the hitting time starting from the states $[H|T]$, $[H|T][]$,
and $[][H|T]$. We select a productive state with probability $(1-p_d)$ and an unproductive one with
probability $p_d$. In either case we traverse some number $0\leq k \leq r-1$ of unproductive
sub-trees, each of which adds $C$ steps. In case of a productive state, the number depends on the
position of $\tau_1$ in the set of $k+1$ elements:
\begin{align*}
    \MHT[{[H|T]}]{A_{i+1}\cup\set{\bot}}=&p_d\cdot \sum_{k=0}^{r-1}\binom{r-1}{k}p_d^{r-1-k}(1-p_d)^k
    \cdot kC & \text{(unproductive)}\\
             &+ (1-p_d)\cdot
             \sum_{k=0}^{r-1}\binom{r-1}{k}\frac{p_d^{r-1-k}(1-p_d)^k}{k+1}
             & \text{(productive)}\\
             &\qquad \cdot \sum_{m=0}^{k}mC + h_{i} + q^{(i)}(m-k)C &
             \text{(\cref{prop:recursion_ds} and \cref{lem:reduce_MHT})}
\end{align*}
Note the denominator $k+1$ accounts for the probability of placing $\tau_i$ at position
$m=0,1,\ldots,k$. At position $m$, there are $mC$ unproductive sub-trees \emph{before} reaching
$\tau_1$ and $(m-k)$ after. In the third line, we split the recursive hitting time using \cref{lem:reduce_MHT}. Note
that the exit state $\Exit_s$ of each recursive sub-tree root $s$ reached in this way is either
$\bot$  (if $k=m$) or the adjacent  unproductive recursive sub-tree root (cf.
\cref{fig:stack_chain}). By \cref{prop:recursion_ds} \cref{prop:recursion_ds:1} we may use
recursion.

Computing the inner sum and canceling out the denominator $k+1$ from the second line, we get:
\begin{align*}
\MHT[{[H|T]}]{A_{i+1}\cup\set{\bot}}=&p_d\cdot \sum_{k=0}^{r-1}\binom{r-1}{k}p_d^{r-1-k}(1-p_d)^k
\cdot kC \\
    &+ (1-p_d)\cdot
    \sum_{k=0}^{r-1}\binom{r-1}{k}p_d^{r-1-k}(1-p_d)^k \cdot \left(\frac{kC}{2} + h_i +
        q^{(i)}\frac{kC}{2}\right) \\
       =& p_d(1-p_d)(r-1)C + (1-p_d)\left(h_i + (1+q^{(i)})\frac{(r-1)(1-p_d)C}{2}\right)
\end{align*}
where the Binomial sums are computed as in the proof of \cref{thm:number_of_tests_ds}. The formulas for
the remaining states are completely analogous, but we have additional steps for the
intermediate detours over $[]$. It is thus convenient to express them in terms of
$\MHT[{[H|T]}]{A_{i+1}\cup\set{\bot}}$:
\begin{align*}
    \MHT[{[][H|T]}]{A_{i+1}\cup\set{\bot}}=& 1+ \MHT[{[H|T]}]{A_{i+1}\cup\set{\bot}}
    \\
    \MHT[{[H|T][]}]{A_{i+1}\cup\set{\bot}}
    =&\MHT[{[H|T]}]{A_{i+1}\cup\set{\bot}} + (1-p_d)q^{(i)}\tag{$\ast$}
\end{align*}
For the second identity, note that visits to $[]$ are counted only in case of failure and
are thus weighted by $q^{(i)}$.

We can now turn to $h_{i+1}$. By \cref{eqn:mean_hitting_times}:
\begin{align*}
    h_{i+1}=& p_d^2\cdot 0 + (1-p_d)p_d + (1-p_d)p_d\cdot \MHT[{[H|T]}]{A_{i+1}\cup\set{\bot}} \\
            & + \frac{(1-p_d)^2}{2}\cdot (\MHT[{[][H|T]}]{A_{i+1}\cup\set{\bot}} +
            \MHT[{[H|T][]}]{A_{i+1}\cup\set{\bot}})
\end{align*}
which, after substitution of ($\ast$) and straightforward simplification becomes
\begin{displaymath}
h_{i+1}= \frac{1}{2} \left(1-p_d^2 + (1-p_d)^3q^{(i)}\right) + (1-p_d)\MHT[{[H|T]}]{A_{i+1}\cup\set{\bot}}
\end{displaymath}
We substitute $\MHT[{[H|T]}]{A_{i+1}}$ and \cref{prop:recursion_ds}
\cref{prop:recursion_ds:2}, then simplify, isolating the terms that are
dependent on $i$:
\begin{multline*}
    h_{i+1} =  (1-p_d)^2h_i -\frac{1+C(r-1)}{2}\cdot(1-p_d)^{2i+4} \\
    + C(1-p_d)^2(r-1) +
    \frac{(1-p)(1+p) + (1-p) ^3}{2}
\end{multline*}

With $\alpha=(1-p_d)^2$, $\beta = -\frac{1+C(r-1)}{2}\cdot (1-p_d)^4$, and
$\gamma=C(1-p_d)^2(r-1) + \frac{(1-p)(1+p) + (1-p)^3}{2}$
we get the following formula (e.g. by iterative substitution):
\begin{equation}
    h_{i+1} = \alpha^{i+1}h_0 + (i+1) \beta\alpha^i + \gamma\,\frac{1-\alpha^{i+1}}{1-\alpha}
    \label{eqn:proof:iterations_ds}
\end{equation}
Its correctness is easily shown by induction on $i\geq 0$.
Substituting $h_0$, $\alpha$, $\beta$, and $\gamma$ into \cref{eqn:proof:iterations_ds} gives:
\begin{multline*}
    h_{i} = (1-p_d)^{2i}h_0 -  i\cdot(1-p_d)^{2(i-1)+4}\cdot \frac{1+C(r-1)}{2} \\+
    \left(C(1-p_d)^2(r-1) + \frac{(1-p)(1+p) + (1-p)^3}{2}\right)\frac{1-(1-p_d)^{2i}}{1-(1-p_d)^2}
\end{multline*}
which simplifies to the desired formula.
\end{proof}
Substituting into \cref{eqn:iterations_ds_full} gives the following theorem:
\begin{theorem}
    \label{thm:hitting_time_ds}
    Let $p_d\in (1-\frac{1}{\sqrt{r}},1)$, $r\in\naturals$. Denote by $C$ the constant from
    \cref{thm:number_of_tests_ds}. Let $\tau$ be a test-case of length $l$ and $A_\tau\subseteq
    \statesDS$ the set of states that output $\tau$. Then:

    \begin{multline*}
    \MHT[\varepsilon]{A_{\tau}}=\frac{1+ \left(\frac{C(1-p_d)^2(r-1)}{p_d(2-p_d)} + \frac{(1-p)(1+p) +
    (1-p)^3}{2p_d(2-p_d)}\right) }{(1-p_d)^{2l+1}} \\
    + \frac{(1+p_d)(1+C)}{2} -\frac{C(1-p_d)(r-1)}{p_d(2-p_d)}
        - \frac{(1+p_d) + (1-p_d)^2}{2p_d(2-p_d)} \\
        - l\cdot (1-p_d)\cdot \frac{1+C(r-1)}{2}
        - (1-p_d)
        +\frac{1}{(1-p_d)}
    \end{multline*}
\end{theorem}

All that remains is to prove \cref{lem:reduce_MHT}:
\begin{proof}[Proof of \cref{lem:reduce_MHT}]
    \begin{align*}
        \MHT[x]{A} 
        =& \sum_{n=0}^{\infty}n\cdot \Pr[\HT[x]{A}=n,\HT[x]{A} < \HT[x]{y}]
        + \sum_{n=0}^{\infty}\sum_{k=0}^{n-1}n\cdot
        \Pr[\HT[x]{A}=n,\HT[x]{A}>\HT[x]{y},\HT[x]{y}=k]
    \end{align*}
    Using Fubini's theorem, the second sum becomes
    \begin{align*}
    &\sum_{k=0}^{\infty}\Pr[\HT[x]{y}=k,\HT[x]{A}>\HT[x]{y}]\sum_{n=k+1}^{\infty}n\cdot
    \Pr[\HT[x]{A}=n|\HT[x]{y}=k,\HT[x]{A}>\HT[x]{y}]\\
    =&
    \sum_{k=0}^{\infty}\Pr[\HT[x]{y}=k,\HT[x]{A}>\HT[x]{y}]\sum_{m=1}^{\infty}(m+k)\Pr[\HT[y]{A}=m]
\end{align*}
and since $\Pr[\HT[y]{A}=0]$, because $y\notin A$, and moreover $\mathcal{M}$ is recurrent (so
$\Pr[\HT[y]{A}<\infty]=1$) this is equal to
\begin{displaymath}
    \sum_{k=0}^{\infty}\Pr[\HT[x]{y}=k,\HT[x]{A}>\HT[x]{y}]\left(k
    + \MHT[y]{A}\right) = \left(\sum_{k=0}^{\infty}\Pr[\HT[x]{y}=k,\HT[x]{A}>\HT[x]{y}]k \right)+ \MHT[y]{A}
    \cdot\Pr[\HT[x]{A}>\HT[x]{y}].
\end{displaymath}
Because $y\notin A$, we have $\HT[x]{A}\neq \HT[x]{y}$ and so
$\Pr[\HT[x]{y}=k,\HT[x]{A}>\HT[x]{y}] +
\Pr[\HT[x]{A}=k,\HT[x]{A}<\HT[x]{y}]=\Pr[\HT[x]{A\cup\set{y}}=k]$. Hence the remaining two
infinite sums add up to $\MHT[x]{A\cup \set{y}}$.
\end{proof}

\section{Evaluation}

\begin{figure}[t]
\centering
\begin{subfigure}{.5\textwidth}
\includegraphics[width=.99\textwidth, clip, trim=0 8pt 40pt 40pt]{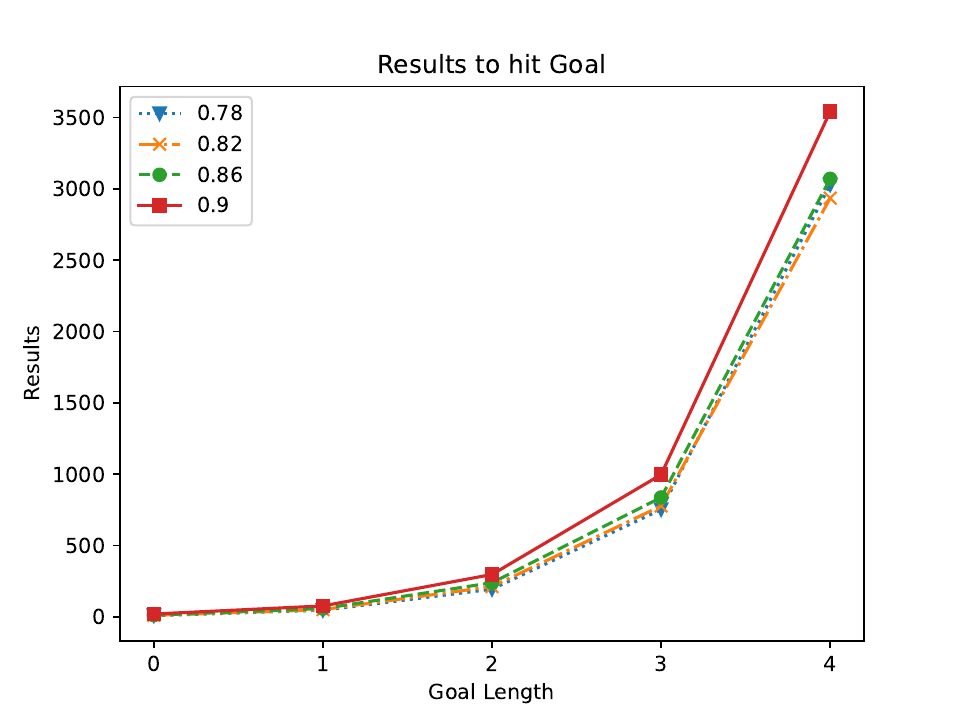}
\caption{Guards: Number Results}
\label{fig:amountOfResultsUntilGoal}
\end{subfigure}%
\begin{subfigure}{.5\textwidth}
\includegraphics[width=.99\textwidth, clip, trim=0 8pt 40pt 40pt]{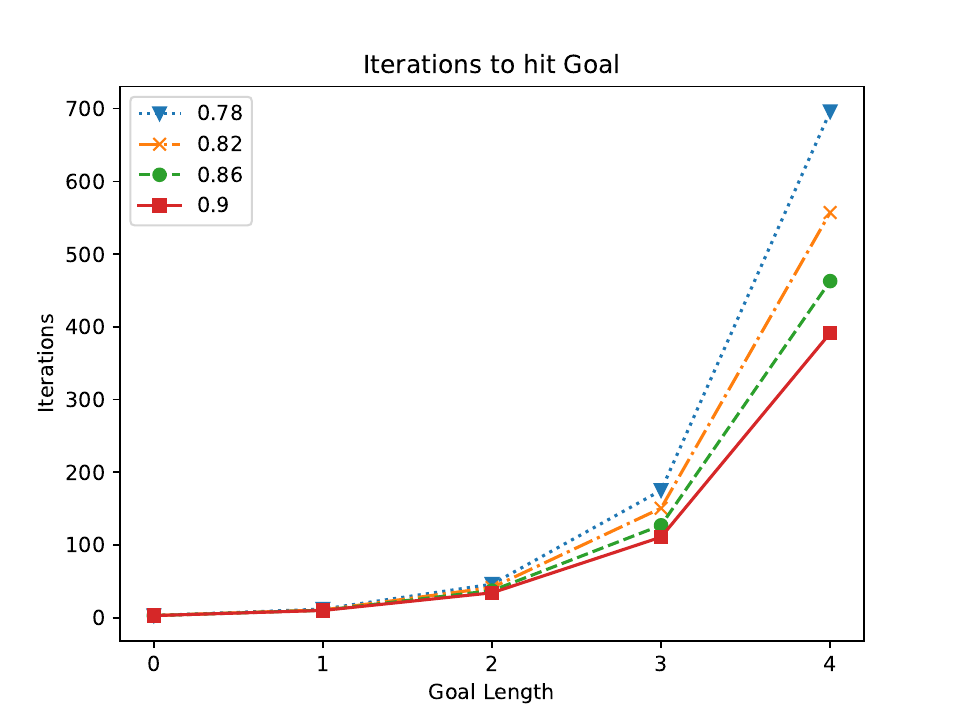}
\caption{Guards: Number Iterations}
\label{fig:amountOfIterationsUntilGoal}
\end{subfigure}\\
\begin{subfigure}{.5\textwidth}
\includegraphics[width=.99\textwidth, clip, trim=0 8pt 40pt 40pt]{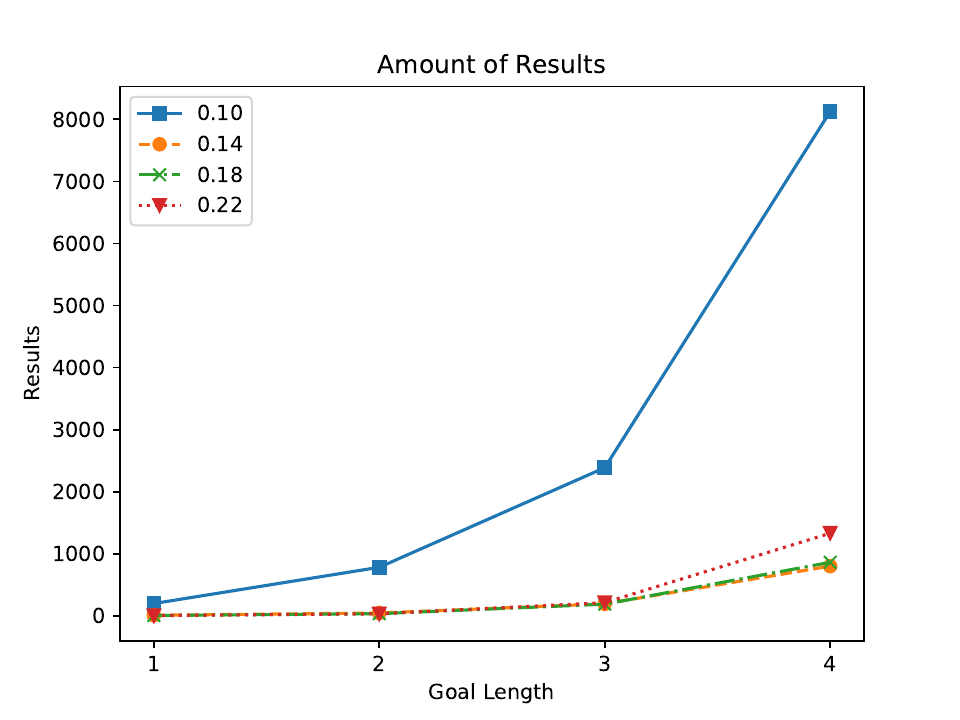}
\caption{Shuffle-and-Drop: Number Results}
\label{fig:goAmountResults}
\end{subfigure}%
\begin{subfigure}{.5\textwidth}
\includegraphics[width=.99\textwidth, clip, trim=0 8pt 40pt 40pt]{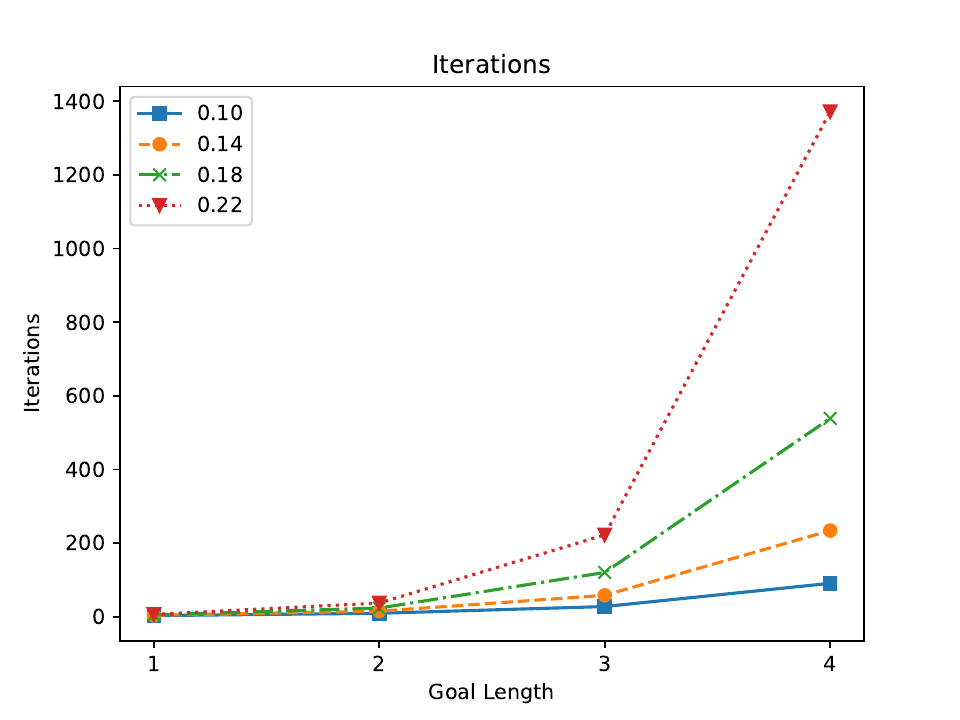}
\caption{Shuffle-and-Drop: Number Iterations}
\label{fig:goIterations}
\end{subfigure}

\caption{Iterations and results until test-case \texttt{[second, \ldots, second]} is reached.}\label{fig:results}
\end{figure}

In the following section, we empirically evaluate the randomization approaches outlined in
\cref{sec:strategies}. We implement the strategy via \emph{guards} using SWI-Prolog \cite{swi}. To
implement the \emph{drop-and-shuffle} strategy, we choose Go-Prolog \cite{go_prolog}. Go-Prolog has a
small and easily modifiable code-base, which simplifies experiments of this kind. We benchmark
these approaches with various choices for the configurable probabilities.

For the benchmarks, we use two programs: a slightly altered
version of the program from \cref{lst:generator_program} and a program generating basic arithmetic
expressions, shown in \cref{lst:arithmetic_expressions}. In both cases, we count the number of
\emph{iterations}, which is defined as the number of times the program needs to be re-run to obtain
the target result, and the \emph{results}, which is the number of outputs of the program before we
obtain the target result.

\subsection{Benchmark 1: Command Sequences}\label{sec:bench1}
The first benchmark is based on \cref{lst:generator_program} with two adjustments We limited the number of
available commands to three. Moreover, each \texttt{command/1}
predicate simply unifies its argument with a corresponding constant (in our case \texttt{first},
\texttt{second}, and \texttt{third}). We executed each benchmarks 1000 times. The results are shown
in \cref{fig:results}.

The \emph{goal length} each benchmark lists on its x-axis is the length of a
list consisting solely of the constant symbol \texttt{second} the respective
number of times. This guarantees that we would not find this test case with the
standard depth-first, left-first search behaviour, but also that is not the
path that would be picked last with depth-first search. For our implementation,
we relied on Janus \cite{andersen2023janus} for SWI as the Python-Prolog bridge
to gather the results.

\paragraph{Guard-Approach Benchmarks:} Every \texttt{command/1} predicate had
an equal probability of $\frac{1}{3}$ for the steady probability.  The
different plots mark different continuation probabilities $p_c$ used for the
respective queries.

\Cref{fig:amountOfResultsUntilGoal} shows the
number of results until a specific target test-case is found. As expected, the number of results
drastically increases with the target list size. Further, only a continuation probability of 0.9
has a higher number of results. \Cref{fig:amountOfIterationsUntilGoal} shows how many iterations
were necessary until the determined target was found. Similar to the number of results, the number
of iterations also grows with increasing list size of the expected outcome. As the continuation
probability increases the total number of iterations decreases.

\paragraph{Drop-and-Shuffle Benchmarks:} As described above, for the Go-Prolog variant we
implemented a drop-probability as discussed in \cref{sec:strategies}.
Otherwise, the benchmarks are still conducted using the same pattern as described
above for increasing goal lengths. \Cref{fig:goAmountResults} shows the
number of produced results whereas \cref{fig:goIterations} shows the
number of iterations.

Note that dropping a clause with probability 0.1 is the same as proceeding to explore it with
probability 0.9. Hence, the probabilities in \cref{fig:goAmountResults,fig:goIterations} are dual to
those above. But note that in this way they \emph{do not fulfill the premise of
\cref{thm:number_of_tests_ds}}: They are below $1-\frac{1}{\sqrt{r}}\approx 0.423$. The mean hitting
time $C$ is thus infinite. And yet, we obtain results! This might seem paradoxical, but is due to the
fact that we return to $\bot$ with probability 1, even if there is no finite mean. Indeed, this
makes these measurements particularly interesting, because they have no defined mean to converge to
--- the law of large numbers applies only if the distribution has finite mean!

Note that the drop-and-shuffle randomization strategy is much more coarse than the guard
strategy by design: The 0.1 drop probability applies to \emph{both} the \texttt{t} predicate as well
as the \texttt{command\{1,2,3\}} predicates. This is quite different from the previous scenario,
where $p_c=0.9\gg p_1=\cdots=p_r=0.33$ were distinct.
Consequently, both the number of iterations and the number of results are notably higher for the
drop-and-shuffle approach: A probability of 0.1 to drop a clause
produces significantly more solutions until a specified goal is found. The drop
probabilities of 0.14 and 0.18 are rather similar for all specified goal lengths. On the other hand,
the number of iterations signals that the number of iterations rises with a higher dropping
probability. In \cref{fig:goAmountResults}, the probability 0.14 outperforms both 0.10 and 0.18.

In \cite{GeLuPe24}, we conjectured that this is due to an inflection point. But the results of
\cref{sec:ds} now show that this reasoning is not verifiable: An undefined function cannot have an
inflection point. It seems impossible to know whether this is an artefact of the inherent randomness
of the measurements (that,  we know, cannot converge to a non-existent mean), or some other effect.
In any case, it underscores that the drop-and-shuffle approach is unwieldy and difficult to analyze.

\subsection{Benchmark 2: Arithmetic Expressions}

In this section, we provide the results of another benchmark to showcase the behavior of both
approaches in a more complicated setting. This time, our goal is to show the
behavior of the two randomization approaches in settings other than those considered in the previous
sections. We use the program in \cref{lst:arithmetic_expressions}, which generates basic
arithmetic expressions build from the two binary operators \texttt{+} and \texttt{\texttimes}, as
well as the unary operator \texttt{-}. Note that only the numbers 1, 2, and 3 can be used within an
expression. Expressions have the form \texttt{[minus, [plus, [1,3]]]}. In
the text below, we use the more readable symbolic representation \texttt{-(+(1,3))} (Polish
notation). The \emph{value} of the previous expression is $-4$.

Our benchmark counts the iterations and number of results (defined as above in \cref{sec:bench1})
until an expression that evaluates to a specific \emph{target value} is found. The target values,
along with example expressions of shortest length, are provided in \cref{tab:sample_expressions}.
As before, we repeat our experiment 1000 times. Clearly, for any target integer value $x$, there are
infinitely many expressions that evaluate to $x$. However, some values can be reached with
relatively simple expressions, whereas others require more complex, nested expressions that can only
be found at lower depths in the SLD-tree.

Note that the results from the preceding sections do not directly extend to the program that we
study here: The program is different (though similar in structure), and we now investigate the
number of steps needed to reach \emph{any} output from an \emph{infinite set} of target outputs. The
second point, in particular, is a major difference from the previous setting. Consider, for example,
the target value 15. Expressions with value 15 are ubiquitous. Suppose during resolution, we
arrive at a partially expanded expression \texttt{+(a,X)}, where \texttt{X} is yet to be derived and
where the value of \texttt{a} is an \emph{arbitrary} integer. Then there are infinitely
many expressions for \texttt{X} that will give the value 15. The only partial expression that cannot
ever be completed to a full expression evaluating to 15 is of the form \texttt{\texttimes(a,X)} for
an expression \texttt{a} with a value that is not a divisor of 15. This means, that we have a much larger
probability of arriving at a \enquote{desired} output than before.

\begin{table}[t]
        \caption{Example expressions of smallest possible size for the target values (in Polish
        notation).}\label{tab:sample_expressions}
        \centering
    \begin{tabular}{ll}
        \toprule
        \textbf{Target Value} & \textbf{Shortest Expression}\\
        \midrule
\;4 & \texttt{+(1,3)}\\
-4 & \texttt{-(+(1,3))}\\
\;6 & \texttt{\texttimes(2,3)}\\
-12 & \texttt{-(\texttimes(3,+(2,2)))}\\
\;15 & \texttt{\texttimes(3,+(2,3))}\\
\bottomrule
    \end{tabular}
\end{table}

\begin{lstlisting}[language={Prolog},basicstyle=\footnotesize\ttfamily,caption={A program generating
arithmetic expressions.}, label={lst:arithmetic_expressions},float]
expr(X) :- const(X).
expr([Operator, Operands]) :- unpack(Operator, Operands).

const(1).
const(2).
const(3).

unpack(plus, [A, B]) :- expr(A), expr(B).
unpack(times, [A, B]) :- expr(A), expr(B).
unpack(minus, [A]) :- expr(A).
\end{lstlisting}

\begin{figure}[t]

\begin{subfigure}{.45\textwidth}
\includegraphics[width=.99\textwidth, clip, trim=1pt 8pt 40pt 40pt]{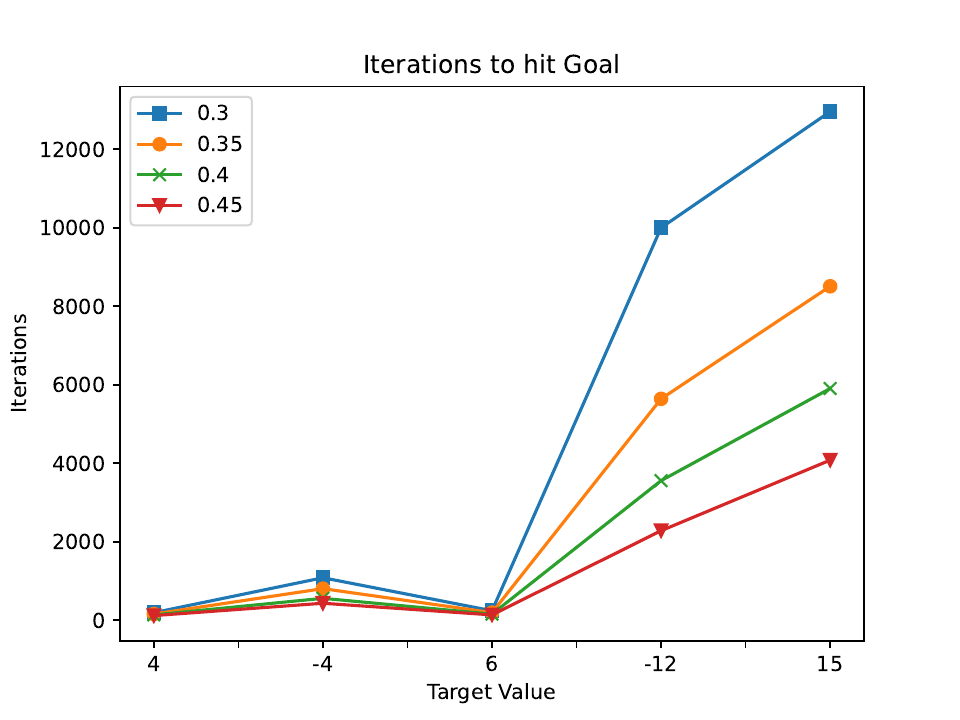}
\caption{Guard: Number of iterations }\label{fig:bench_expr:guard_it}
\end{subfigure}
\hspace{.5em}
\begin{subfigure}{.45\textwidth}
\includegraphics[width=.99\textwidth, clip, trim=8pt 8pt 40pt 40pt]{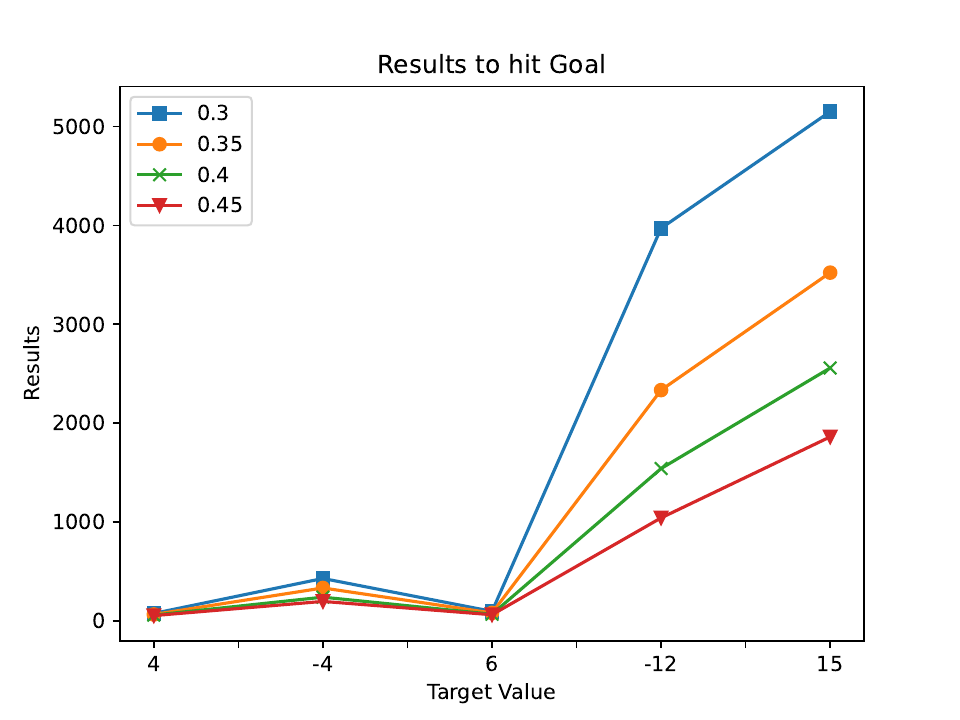}
\caption{Guard: Number of results}\label{fig:bench_expr:guard_res}
\end{subfigure}

\begin{subfigure}{.45\textwidth}
\includegraphics[width=.99\textwidth, clip, trim=1pt 8pt 40pt 40pt]{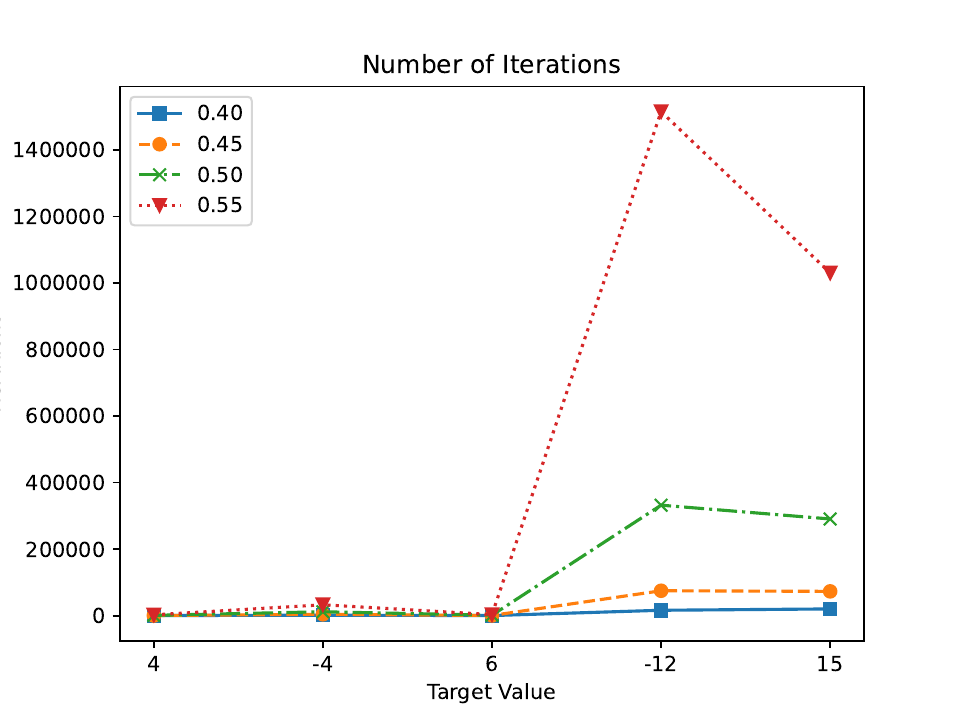}
\caption{Shuffle-and-Drop: Number of iterations }\label{fig:bench_expr:shuffle_it}
\end{subfigure}
\hspace{.5em}
\begin{subfigure}{.45\textwidth}
\includegraphics[width=.99\textwidth, clip, trim=8pt 8pt 40pt 40pt]{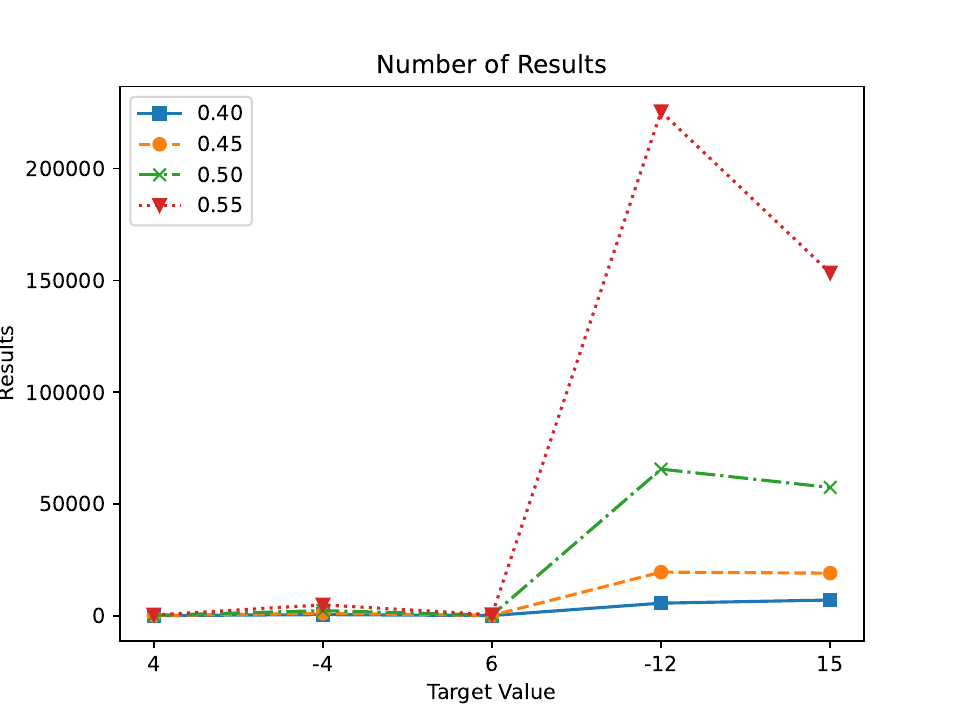}
\caption{Shuffle-and-Drop: Number of results}\label{fig:bench_expr:shuffle_res}
\end{subfigure}

    \caption{Iterations and results until the generated expression has a specific value}
\label{fig:bench_expr}
\end{figure}

\paragraph{Guard-Approach Benchmarks:}
For the guard approach, we used the continuation probabilities $p_c=0.3$ to $p_c=0.45$ in increments
of $0.05$. The continuation probability $p_c$ was used only for the clause in line 2 of
\cref{lst:arithmetic_expressions}. All other clauses, including the facts in lines 4-6, were guarded
with a fixed probability of $0.33$.

The results of the guard-approach benchmark is shown in
\cref{fig:bench_expr:guard_it,fig:bench_expr:guard_res}. The first observation is that lower
continuation-probabilities result in a larger number of results until a target clause is reached. This is in
contrast to the previous benchmark (cf. \cref{sec:bench1}), where higher probabilities lead to a large number of results.
This is likely because we now try reach any one expression from an infinite set of expressions. If
the continuation probability is higher, we reach deeper into the SLD tree. Unlike before, however,
we are quite likely to find our target that way.

Note that expressions with a negative sign require both more iterations and produce more unwanted
results, before being reached. Target values requiring more complex expressions take longer to
reach, as expected. There is significant increase in the time required to produce a target result,
if three or more sub-expressions are needed (i.e.\ for -12 and 15). Note also that there is a small
peak for expression -4. This is likely, because $-4$ requires on additional operator if compared
with $4=2+2$ and $6=2\cdot 3$.

\paragraph{Drop-and-Shuffle Benchmarks:}

We used a drop-probability $p$ ranging from $0.4$ to $0.55$ (in increments of $0.05$). Values
of $p> 0.55$ took excessively longer to benchmark, and we could not complete 1000 test-runs within
four days of running the benchmark for such values of $p$. The results of the benchmarks are shown
in \cref{fig:bench_expr:shuffle_it,fig:bench_expr:shuffle_res}.

First, we note that the numbers are orders of magnitude larger than for the guard approach.
While the numbers are difficult to compare (unlike \cref{sec:bench1} probabilities are not
completely dual), there appears to be a significant increase in runtime and uninteresting results
that do not evaluate to the target value. As we have noted before, the shuffle-and-drop strategy is
more coarse and the drop probability affects all clauses, not just the one in line 2 of
\cref{lst:arithmetic_expressions}. This is the most likely explanation for the excessive runtime.

Next, note that the metrics for target value 15 are consistently lower than for target value -12. We
do not observe this effect for the guard approach in
\cref{fig:bench_expr:guard_it,fig:bench_expr:guard_res}. In absence of mathematical rigour, we can,
again, only conjecture as to the cause. There are two ways of reading this result: It shows that 15
is intrinsically easier to reach in the shuffle-and-drop approach than -12, or that -12 is
intrinsically harder to reach, when using the shuffle-and-drop approach. We conjecture that the second
interpretation is correct. The lowest test-case that can produce -12 is at a greater depth than 15
(see \cref{tab:sample_expressions}). Moreover, the number -12 it requires even greater depth to
reach if the outermost operator is not \enquote*{\texttt{-}}. Looking back to
\cref{thm:hitting_time_ds} in \cref{subsec:ds:infinite}, we see that the runtime grows exponentially
in \emph{twice} the depth of the test-case. This result does not extend to our present setting, but
it seems likely that the runtime must grow at least exponentially in the depth of the
\emph{shortest} possible derivation of the desired output, governed by a similar growth-function to
that given in \cref{thm:hitting_time_ds}. If that is the case, reaching -12 becomes much more difficult
than reaching 15. While a similar argument would seem to apply to the guard approach as well, where
we don't see this effect, the base of the exponential is smaller and the exponent is not scaled by
two (see \cref{cor:hitting_time_guard_theta}). It is quite possible that the relatively short
expressions yielding -12 are dominated by other factors in the guard approach setting, due to the
smaller base and exponent.

\section{Conclusion}

We have presented two approaches to randomize the SLD derivation of test-cases in Prolog and studied
their performance in terms of expected time to hit a test-case, and mean number of test-cases produced.
To this end, we presented a detailed analysis of the random behavior of test-case generation using
Prolog and Markov chains. Our theorems allow a precise calibration of the probabilities to adjust
the expected number of test-cases per query. When looping on such a query, the rate of growth of the
mean-hitting time for a given test-case is exponential in its depth, where the base is the product
of the involved probabilities. We then compared both strategies and various sets of values for the
involved probabilities empirically. We find that the guard approach that uses an unmodified Prolog
implementation provides a very fine-grained control over the randomization and thus produces
test-cases quicker.

In future work, we plan to study the semantics of this approach when negation-as-failure is
involved. In particular, randomization may lead to a false refutation of \texttt{q(t\_1,\ldots,t\_k)}
in the goal \texttt{\textbackslash+ q(t\_1,\ldots, t\_k)}. However, this may be acceptable, if it
occurs with low probability. In a similar vein, the treatment of negation as failure might require
randomization strategies entirely different from those we have presented here, which is another
interesting topic for future research.

\bibliographystyle{tlplike}
\bibliography{bibliography}

\end{document}